\documentclass[12pt,onecolumn,journal]{IEEEtran}
\usepackage{latexsym}
\usepackage[margin=2.54cm, lmargin=2.54cm]{geometry}
\usepackage{amsfonts}
\usepackage{amsbsy}
\usepackage{amsmath,amssymb}
\usepackage{amsthm}
\usepackage{times}
\usepackage{graphicx}
\usepackage{enumerate}
\usepackage[usenames]{color}
\usepackage[dvips]{pstcol}
\usepackage{epstopdf}
\usepackage{cite}
\usepackage{amsmath}
\usepackage{amssymb}
\usepackage{amsfonts}
\usepackage{graphicx}
\usepackage{epsfig}
\usepackage{psfrag}
\usepackage{xcolor}
\usepackage{amsfonts, bm}
\usepackage{epstopdf}
\usepackage{cite}
\usepackage{color}
\usepackage{xcolor}
\usepackage{verbatim}
\usepackage{subfig}
\usepackage{algorithm}

\newtheorem{theorem}{Theorem}

\input epsf
\begin{document}
\title{Intelligent Reflecting Surface Aided Full-Duplex Communication: Passive Beamforming and Deployment Design }


\author{Yunlong Cai, \IEEEmembership{Senior Member,~IEEE,} Ming-Min Zhao, \IEEEmembership{Member,~IEEE,}  Kaidi Xu,  and Rui Zhang, \IEEEmembership{Fellow,~IEEE}
\thanks{
Y. Cai,  M. Zhao, and K. Xu   are with the College of Information Science and Electronic Engineering, Zhejiang University, China (e-mail: ylcai@zju.edu.cn; zmmblack@zju.edu.cn; xukaidi13@126.com). R. Zhang is with the Department of Electrical and Computer Engineering, National University of Singapore, Singapore (e-mail: elezhang@nus.edu.sg).
}
}

\maketitle
\vspace{-2.8em}
\begin{abstract}
 This paper investigates the passive beamforming and deployment design for an intelligent reflecting surface (IRS) aided full-duplex (FD) wireless system, where an FD access point (AP) communicates with an uplink (UL) user and a downlink (DL) user simultaneously over the same time-frequency dimension with the help of  IRS. Under this setup, we consider three deployment cases: 1) two distributed IRSs placed near the UL user and DL user, respectively; 2) one centralized IRS placed near the DL user;  3) one centralized IRS placed near the UL user. In each case, we aim to minimize the weighted sum  transmit power consumption of the AP and UL user by jointly optimizing their transmit power  and the passive reflection coefficients at the IRS (or IRSs), subject to the UL and DL users' rate constraints and the uni-modulus constraints on the IRS reflection coefficients.
First, we analyze the minimum transmit power required  in the IRS-aided FD  system under each deployment scheme, and compare it with that of the corresponding half-duplex (HD)  system. We show that the FD system outperforms its HD counterpart for all IRS deployment schemes, while the distributed  deployment further outperforms the other two centralized deployment schemes. Next, we  transform the challenging power minimization problem into an equivalent but more tractable form and propose an efficient algorithm to solve it based on the block coordinate descent (BCD) method. Finally, numerical results are presented to validate our analysis as well as the efficacy  of the proposed passive beamforming design. 
\end{abstract}

\IEEEpeerreviewmaketitle

\section{Introduction}
\label{sec:intro}

Recently, intelligent reflecting surface (IRS) and its various equivalents have emerged as a promising technology to enhance the spectral efficiency of wireless communication systems with low hardware cost and energy consumption \cite{Liaskos2018mag, Renzo2019, Basar2019, WuTutorial2020, Wu2019Magazine}. Specifically, IRS is generally equipped with a planar surface composed of a large number of passive reflecting elements, each of which can induce an independent phase shift and/or amplitude change of the incident signal in  real time. Based on the  channel state information (CSI) and with the aid of a smart controller, IRS is able to modify/reconfigure the signal propagation by dynamically adjusting its reflection coefficients such that the desired  and interfering signals can be added constructively and destructively at the receivers, respectively,  to boost the desired signal power and/or suppress the co-channel interference (CCI), thus achieving communication performance improvement. Additionally, since such passive reflecting elements do not require any active transmit radio frequency (RF) chains (which constitute  high-cost amplifiers, filters, mixers, attenuators and detectors, etc.), their energy and hardware costs are much lower than those of the active components in traditional base stations (BSs), access points (APs), and relays. As a result, IRSs can be flexibly deployed in wireless networks and seamlessly  integrated into the existing cellular or WiFi systems,  with controllable interference to each other as they usually have much smaller signal coverage than active BSs/APs/relays. Due to the above advantages,  IRS  has been extensively  studied under various setups (see, e.g., \cite{QwuTWC2019, Huang2019, zhao2019intelligent, Jiang2019, Yang2019, zuo2020resource}), where the effectiveness of IRS in enhancing these systems' performance was demonstrated.

To boost the spectral efficiency of wireless  systems, the full-duplex (FD) communication  is another  promising technique \cite{Sabharwal2014, Zheng2015, Riihonen2011}. Compared with the traditional half-duplex (HD) system, it  better utilizes the spectrum by enabling signal transmission and reception over the same time-frequency dimension and thus can double  the spectral efficiency theoretically. Although the FD system  offers promising  spectral efficiency gains, it also brings   challenges in practice, such as the self-interference (SI) and CCI caused by the simultaneous downlink (DL) and uplink (UL) transmissions. These interferences, if left unattended, can encroach the gain offered by FD and even  degrade the system performance as compared to traditional HD. Fortunately, many efficient SI cancellation techniques have been proposed in the literature \cite{Riihonen2011, Everett2014, Zhang2015Mag}, such as passive suppression, analog and digital cancellations, etc., and it has been reported that the SI can be suppressed up to the background noise floor at an FD node \cite{Dinesh2013}, which leads to implementable FD systems. However, the CCI inherited from the FD operation can also  cause significant spectral efficiency degradation, despite the fact that the SI can be effectively suppressed. In practice, the interference from the UL users to the DL users (UL-to-DL interference) can be  detrimental if the UL and DL users are located close to each other. To avoid  the performance degradation due to CCI, various methods such as DL/UL user pairing and transmission scheduling  have been proposed in the literature (see e.g., \cite{Xin2015, Park2019CCI} and the references therein).

Motivated by the above, in this paper, we combine the two  spectral-efficient wireless techniques, i.e., IRS and FD into a new system, which can leverage  IRS's  interference cancellation and signal enhancement capabilities to further improve the FD communication  performance. 

\subsection{Prior Works}
IRS has been studied recently in various wireless systems, such as  IRS-aided orthogonal frequency division multiplexing (OFDM) \cite{Yang2019},  mmWave communication \cite{Zuo2020CL, Wang2020TVT}, physical layer security \cite{Cui2019, Hong2020TCOM}, wireless power transfer \cite{wu2019joint, PanJSAC2020}, etc. To the best of our knowledge, there are only few works that have
studied IRS-aided FD systems \cite{Zhang2020CL, Shen2020full-duplex, Xu2020full-duplex}. Specifically, in \cite{Zhang2020CL}, the authors studied an IRS-aided FD multiple-input multiple-output (MIMO) two-way communication system and the sum-rate maximization problem was considered by jointly optimizing the active precoders at the sources and the IRS reflection coefficients through the Arimoto-Blahut algorithm. Under the same system model, the authors in \cite{Shen2020full-duplex} further proposed an alternating optimization algorithm to solve the sum-rate maximization problem with lower computational complexity. In \cite{Xu2020full-duplex}, an IRS-assisted FD cognitive radio system was investigated, where the IRS is employed to enhance the performance of the secondary network and in the meantime mitigate the interference caused to the primary users. Note that these existing works on IRS-aided FD systems mainly focused on joint active and passive beamforming design, while
the fundamental advantages of IRS-aided FD versus HD system are not fully characterized yet. 
Moreover, with a given number of IRS reflecting elements, there are various IRS deployment strategies for an IRS-aided FD system, e.g., 
placing multiple IRSs in the network for far-apart  users separately, or forming them as one large IRS and placing it near a cluster of nearby  users, which are referred to as distributed and centralized deployment, respectively. 
The IRS deployment problem has been investigated in the point-to-point channel  \cite{Han2020WCL} and the UL (DL) multiple access (broadcast) channel  \cite{Zhang2020deploy}; while  it is yet unclear which of the above two IRS deployment strategies (centralized or distributed) achieves better performance in IRS-aided FD systems.

\subsection{Main Contributions}
In this paper, we consider an IRS-aided FD system where an FD AP (equipped with one transmit antenna and one receive antenna) communicates with a single-antenna UL user and a single-antenna DL user simultaneously over the same time-frequency dimension  with the help of a given number of  IRS reflecting elements. For simplicity, we assume that the IRS (or IRSs if deployed in a distributed manner) is placed in close vicinity to the users to minimize the path loss. Under this setup, we investigate three deployment cases: 1) two distributed IRSs placed near the UL user and DL user, respectively; 2) one centralized IRS placed near the DL user; and 3) one centralized IRS placed near the UL user. In each case, we aim to minimize the weighted sum transmit power consumption of the AP and UL user by jointly optimizing their  transmit power  and the passive reflection coefficients at the IRS (or IRSs), subject to the UL and DL users' rate constraints and the uni-modulus constraints on the IRS reflection coefficients. To focus on the deployment strategy design and passive beamforming optimization, we assume for simplicity the availability of  CSI for all  channels involved, which can be obtained by various existing
channel estimation methods (e.g., \cite{zheng2019intelligent, Yang2019, He2019_CE}). The main contributions of this paper in view of the
existing literature are summarized as follows.

\begin{itemize}
\item First, by assuming Rayleigh fading channels and with an asymptotically large number of IRS reflecting elements, we analyze the minimum  transmit power required  in the above-mentioned three IRS deployment cases. In particular, we show that with the aid of IRS, the FD system always outperforms the HD system regardless of the UL-to-DL interference. This is in sharp contrast with the conventional system without IRS, where the FD operation is not always beneficial, especially in the case that the UL-to-DL interference is severe. Besides, we show  that  the minimum power consumption  with the distributed deployment (Case 1) is much lower than that  with the centralized deployment (Case 2 or 3) in the IRS-aided FD system, which unveils  that the distributed  deployment generally outperforms  centralized deployment  in terms of  power consumption.

\item Second, for arbitrary  channels  and finite  number of reflecting elements, we propose an efficient passive beamforming design algorithm for solving the formulated power minimization problem, which is difficult to solve due to the  non-convex  objective function and uni-modulus constraints. Specifically, we first transform the original problem into an equivalent but more tractable form by introducing an auxiliary variable, and then propose an algorithm for solving it  by employing the block coordinate descent (BCD) method. The computational complexity and convergence property of the proposed algorithm are analyzed.

\item Finally, numerical results are presented to validate our analysis and show  the performance of the proposed algorithm. Particularly, we draw useful insights into the impact of the distance between the UL and DL users, the number of reflecting elements and the UL/DL target rates on the total transmit power. Besides, performance comparison between the proposed algorithm and a low-complexity heuristic algorithm is provided.
\end{itemize}

\subsection{Organization}

The rest of this  paper is organized  as follows. Section~\ref{Section2:system} describes the system model and formulates the optimization
problem of interest. In Section~\ref{Section3:analysis},  we  analyze  the minimum power consumption for FD and HD systems under different IRS deployment cases. In Section~\ref{Section4:passivebeamformingalgorithm}, we propose a new algorithm to optimize the passive beamforming for solving the formulated problems. The simulation results are presented in Section~\ref{Section5:simulations} and our conclusions are drawn in Section~\ref{Section6:conclusion}.

\emph{Notations:} Scalars, vectors and matrices are respectively
denoted by lower case, boldface lower case and boldface upper case
letters.  $\mathbf{I}$ represents an identity matrix and $\mathbf{0}$
denotes an all-zero matrix.  
 For a matrix $\mathbf{A}$,
${{\bf{A}}^T}$, $\text{conj}(\mathbf{A})$, ${{\bf{A}}^H}$ and $\|\mathbf{A}\|$
denote its transpose, conjugate, conjugate transpose and Frobenius
norm, respectively.
 $\text{diag}(\mathbf{A})$ denotes a vector whose elements are the corresponding ones on the main diagonal of $\mathbf{A}$.
For a vector $\mathbf{a}$, $\text{Diag}(\mathbf{a})$ denotes a diagonal matrix with each diagonal element being the corresponding element in $\mathbf{a}$.
$\mathbf{1}$ denotes an all-ones vector. 
$\Re\{\cdot \}$ ($\Im\{\cdot\}$) denotes the real (imaginary) part of a
variable, while $| \cdot |$ represents the absolute value of a complex
scalar.  ${\mathbb{C}^{m \times n}}\;({\mathbb{R}^{m \times n}})$
denotes the space of ${m \times n}$ complex (real) matrices.  The letter $j$ is used to represent $\sqrt{-1}$ when there is no ambiguity.
$\mathbb{Z}^+$ denotes a set of positive integers.
The operator $\angle$ takes the phase angles of the elements in a
matrix. $\mathbf{a}\sim \mathcal{CN}(\mathbf{0}, \mathbf{A})$ denotes that  the random vector   $\mathbf{a}$
follows the circularly-symmetric  complex Gaussian distribution with zero mean and covariance matrix $\mathbf{A}$.

\section{System model and problem formulation}
\label{Section2:system}

In this section, we introduce the system model and  problem formulation for three different IRS deployment cases under investigation.
\subsection{FD System} 
Consider an FD system consisting of an AP,  a UL user and a DL user. The AP is equipped with a single transmit antenna and a receive antenna, and operates in the FD mode. 
The UL and DL users are both equipped with a single
 antenna, and operate in the HD mode.\footnote{In order to focus on our study, we consider a typical single-antenna two-user FD system; while the results in this paper can be extended to the general setup of  IRS-aided  multi-antenna and/or multi-user FD systems, which are left for future work.}

  Generally, IRS is deployed  near the  users for enhancing their performance.
 Assuming a total of $N\in\mathbb{Z}^+$ IRS reflecting elements utilized  in the  system, 
for this case  depicted in Fig.~\ref{fig:structure1} (a),
 there are two IRSs, namely    IRS 1 and IRS 2, which are   equipped with $\rho N\in\mathbb{Z}^+$ and $(1-\rho) N\in\mathbb{Z}^+$  reflecting elements and deployed near the UL/DL user, respectively, where  $0<\rho<1$ denotes a preset ratio.
 Since  IRS 1 and IRS 2  are   sufficiently far  from the DL user and the UL user, respectively, the link that IRS 1 reflects the signals from the AP to the DL user
 as well as  that IRS 2 reflects the signals from the UL user  to the AP can be neglected, as compared to the other links shown in the figure.
The received signal  at the AP  is expressed as
 \begin{equation}
 y_{A}=h_{AU} \sqrt{p_{U}} s_{U}+\mathbf{f}^H_{AI}\mathbf{\Theta}_{U}\mathbf{f}_{IU}\sqrt{p_{U}} s_{U}+n_{A},
 \end{equation}
where $s_{U}$ denotes the transmit signal of the UL user, $p_{U}>0$ denotes the transmit power  of the UL user, and
$n_{A}\in\mathbb{C}$ denotes the independent and identically distributed (i.i.d.)  complex  Gaussian noise  at the AP with zero mean and variance $\sigma^2_{A}$.
$\mathbf{f}_{IU}\in\mathbb{C}^{\rho N\times 1}$ and $\mathbf{f}_{AI}\in\mathbb{C}^{\rho N\times 1}$ denote respectively  the channel vector between the UL user and  IRS 1, and that  between  IRS 1 and the AP.
$h_{AU}\in\mathbb{C}$ denotes the channel coefficient between the UL user and the AP.
$\mathbf{\Theta}_{U}\in\mathbb{C}^{\rho N\times \rho N}$ denotes the passive beamforming matrix at  IRS 1 placed near the UL user, which is a diagonal matrix  due to no signal coupling/joint processing over its passive reflecting elements.
Since the CSI of the SI link can be obtained at the AP, based on certain interference cancellation techniques \cite{Haija2016TCOM, Cai2020FD}, we assume that the  SI at the AP can be eliminated completely for the sake of exposition.

On the other hand, the  received signal at the DL user  is given by
 \begin{equation}
 y_{D}=h_{DA}\sqrt{p_{A}}s_{D}+\mathbf{g}^H_{DI}\mathbf{\Theta}_{D}\mathbf{g}_{IA}\sqrt{p_{A}}s_{D}+\underbrace{g\sqrt{p_{U}}s_{U}+\mathbf{f}^H_{DI}\mathbf{\Theta}_{U}\mathbf{f}_{IU} \sqrt{p_{U}}s_{U}+\mathbf{g}^H_{DI}\mathbf{\Theta}_{D}\mathbf{g}_{IU} \sqrt{p_{U}}s_{U}}_{\textrm{interference from the UL user}}+n_{D},
 \end{equation}
 where    $s_{D}$ and  $p_{A}>0$  denote the transmit signal for the DL user and    the transmit power at the AP, respectively.
  $\mathbf{g}_{IA}\in\mathbb{C}^{(1-\rho)N\times 1}$,  $\mathbf{g}_{DI}\in\mathbb{C}^{(1-\rho)N\times 1}$,  $\mathbf{f}_{DI}\in\mathbb{C}^{\rho N\times 1}$ and $\mathbf{g}_{IU}\in\mathbb{C}^{(1-\rho)N\times 1}$ denote the channel vectors between the AP and  IRS 2,  between  IRS 2 and the DL user,   between  IRS 1 and the DL user, and  between the UL user and IRS 2, respectively.
  $h_{DA}\in\mathbb{C}$ denotes the channel coefficient between the AP and the DL user. $\mathbf{\Theta}_{D}\in\mathbb{C}^{(1-\rho)N\times (1-\rho)N}$ denotes the diagonal beamforming matrix at  IRS 2 which is placed near the DL user.
  $g$ denotes the channel coefficient between the UL user and the DL user. $n_{D}\in\mathbb{C}$ denotes the i.i.d. complex  Gaussian noise  at the DL user with zero mean and variance $\sigma^2_{D}$. The link between the two IRSs is also  neglected due to their large distance and thus  the high path loss.

\begin{figure*}[!t]
\centering
\scalebox{0.46}{\includegraphics{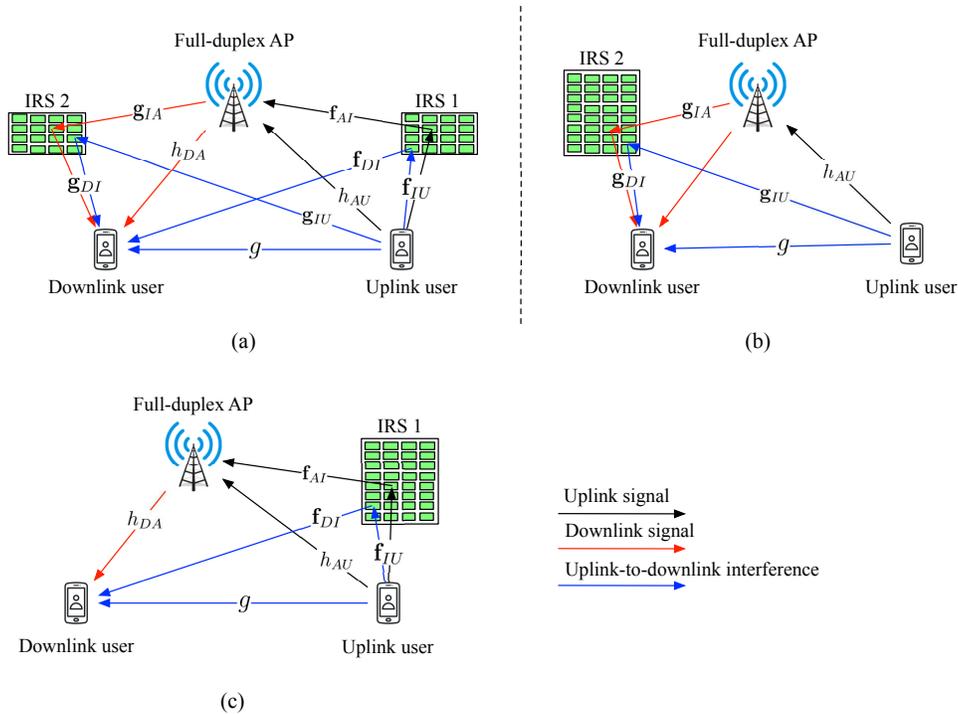}}
\caption{ (a) Case 1:   IRSs placed near both users; (b) Case 2: IRS placed  near the downlink user only; (c) Case 3: IRS  placed near the uplink user only.  }\label{fig:structure1}
\end{figure*}

Let us define the effective channel power gains, namely,   $\lambda_D(\boldsymbol{\theta}_D)\triangleq |h_{DA}+\mathbf{g}^H_{DA}\boldsymbol{\theta}_{D}|^2$, where $\mathbf{g}_{DA}\triangleq \textrm{Diag}(\textrm{conj}(\mathbf{g}_{IA}))\mathbf{g}_{DI}$ and $\boldsymbol{\theta}_{D} \triangleq\textrm{diag}(\boldsymbol{\Theta}_{D})$,
$\lambda_{U}(\boldsymbol{\theta}_{U})\triangleq |h_{AU}+\mathbf{f}^H_{AU}\boldsymbol{\theta}_{U}|^2$, where $\mathbf{f}_{AU}\triangleq \textrm{Diag}(\textrm{conj}(\mathbf{f}_{IU}))\mathbf{f}_{AI}$ and $\boldsymbol{\theta}_{U} \triangleq\textrm{diag}(\boldsymbol{\Theta}_{U})$,
and $\lambda_{DU}(\boldsymbol{\theta}_{U}, \boldsymbol{\theta}_{D})\triangleq |g+\mathbf{d}^H\boldsymbol{\theta}_{DU}|^2$, where $\boldsymbol{\theta}_{DU}\triangleq[\boldsymbol{\theta}^H_{U}, \boldsymbol{\theta}^H_{D}]^H$, $\mathbf{d}\triangleq [\mathbf{f}^H_{DU}, \mathbf{g}^H_{DU}]^H$, $\mathbf{f}_{DU}\triangleq\textrm{Diag}(\textrm{conj}(\mathbf{f}_{IU}))\mathbf{f}_{DI}$, and $\mathbf{g}_{DU}\triangleq\textrm{Diag}(\textrm{conj}(\mathbf{g}_{IU}))\mathbf{g}_{DI}$. We aim to minimize the   weighted sum transmit power consumption of the UL user and AP by jointly optimizing their corresponding  transmit power levels $p_U$ and $p_A$, and the  passive beamforming vectors $\boldsymbol{\theta}_{U}$ and $\boldsymbol{\theta}_{D}$ at the two  IRSs, subject to the UL and DL users' rate constraints and the  uni-modulus constraints on the elements of the IRS passive beamforming vector. Accordingly, the  optimization problem is formulated as
 \begin{subequations} \label{eq:case3}
 \begin{align}
 \min_{p_{U}, p_{A}, \boldsymbol{\theta}_{U}, \boldsymbol{\theta}_{D} }\quad & \kappa_{U} p_{U}+ \kappa_{A} p_{A}  \label{objective22}\\
 \mbox{s.t.}\quad 
 & \log\left(1+\frac{p_{U}\lambda_{U}(\boldsymbol{\theta}_{U})}{\sigma^2_{A}}\right)\geq \gamma_{A},\label{case3:b}\\
 & \log\left(1+\frac{p_{A}\lambda_{D}(\boldsymbol{\theta}_{D})}{p_{U} \lambda_{DU}(\boldsymbol{\theta}_{U}, \boldsymbol{\theta}_{D})+\sigma^2_{D}}\right)\geq \gamma_{D}, \label{case3:c}\\ 
  & |\boldsymbol{\theta}_{U}(n_1)|=1, \, |\boldsymbol{\theta}_{D}(n_2)|=1, \, \forall  n_1\in\{1,\ldots, \rho N\},\, \forall n_2\in\{1,\ldots, (1-\rho)N\}, \label{constantmodulus2223}
 \end{align}
 \end{subequations}
where we have assumed base-2 for the logarithm function,  $\kappa_{U}>0$ and $\kappa_{A}>0$ represent the power  weights corresponding to the UL user and AP, respectively.   Constraints \eqref{case3:b}  and  \eqref{case3:c}  correspond to  the rate requirements for the UL and DL transmissions, respectively, where $\gamma_{A}$ and $\gamma_{D}$ denote the target rates in bits per second per Hertz (bits/s/Hz).
   \eqref{constantmodulus2223} denotes the uni-modulus constraints on all  elements of the IRS passive beamforming vector.

Next, we consider two special cases of the general IRS deployment introduced in the above. First, for the case of placing all reflecting elements to  IRS 2 near the DL user  (i.e., with $\rho=0$) as shown in Fig.~\ref{fig:structure1} (b), let $\boldsymbol{\theta}_{DU} \in\mathbb{C}^{N\times 1}$ denote the entire passive beamforming vector at  IRS 2. By defining
$\tilde{\mathbf{g}}_{IA}\in\mathbb{C}^{N\times 1}$,  $\tilde{\mathbf{g}}_{DI}\in\mathbb{C}^{N\times 1}$ and $\tilde{\mathbf{g}}_{IU}\in\mathbb{C}^{N\times 1}$ as the channel vectors between the AP and  IRS 2,  between  IRS 2 and the DL user, and  between the UL user and  IRS 2, respectively,
 the optimization problem in this case can be   formulated  by simplifying  \eqref{eq:case3} as
 \begin{subequations} \label{eq:case1}
 \begin{align}
 \min_{p_{U}, p_{A}, \boldsymbol{\theta}_{DU} }\quad & \kappa_{U} p_{U}+ \kappa_{A} p_{A}  \label{objective22}\\
 \mbox{s.t.}\quad 
 & \log\left(1+\frac{p_{U}|h_{AU}|^2}{\sigma^2_{A}}\right)\geq \gamma_{A},\label{transmitpower222}\\
 & \log\left(1+\frac{p_{A}\omega_{D}(\boldsymbol{\theta}_{DU})}{p_{U} \omega_{DU}(\boldsymbol{\theta}_{DU})+\sigma^2_{D}}\right)\geq \gamma_{D}, \label{transmitpower222233}\\ 
  & |\boldsymbol{\theta}_{DU}(n)|=1, \forall  n\in\{1,\ldots, N\}. \label{constantmodulus22211}
 \end{align}
 \end{subequations}
 Here we have $\omega_{D}(\boldsymbol{\theta}_{DU})\triangleq |h_{DA}+\tilde{\mathbf{g}}^H_{DA}\boldsymbol{\theta}_{DU}|^2$, where $\tilde{\mathbf{g}}_{DA}\triangleq \textrm{Diag}(\textrm{conj}(\tilde{\mathbf{g}}_{IA}))\tilde{\mathbf{g}}_{DI}$, 
 and $\omega_{DU}(\boldsymbol{\theta}_{DU})\triangleq |g+\tilde{\mathbf{g}}^H_{DU}\boldsymbol{\theta}_{DU}|^2$, where $\tilde{\mathbf{g}}_{DU}\triangleq \textrm{Diag}(\textrm{conj}(\tilde{\mathbf{g}}_{IU}))\tilde{\mathbf{g}}_{DI}$.

Second, for  the other  case of placing all reflecting elements to IRS 1  near the UL user (i.e., with $\rho=1$) as  depicted in Fig.~\ref{fig:structure1} (c), let us define   $\tilde{\mathbf{f}}_{IU}\in\mathbb{C}^{N\times 1}$ and $\tilde{\mathbf{f}}_{AI}\in\mathbb{C}^{N\times 1}$ as the channel vectors between the UL user and  IRS 1  and  between  IRS 1 and the AP, respectively, and $\tilde{\mathbf{f}}_{DI}\in\mathbb{C}^{N\times 1}$  as the channel vector between  IRS 1 and the DL user. Similarly, the optimization problem in this case can be formulated as
 \begin{subequations} \label{eq:case2}
 \begin{align}
 \min_{p_{U}, p_{A}, \boldsymbol{\theta}_{DU} }\quad & \kappa_{U} p_{U}+\kappa_{A} p_{A}  \label{objective22}\\
 \mbox{s.t.}\quad
 & \log\left(1+\frac{p_{U}\xi_{U}(\boldsymbol{\theta}_{DU})}{\sigma^2_{A}}\right)\geq \gamma_{A},\label{transmitpower2222}\\
 & \log\left(1+\frac{p_{A} |h_{DA}|^2}{p_{U} \xi_{DU}(\boldsymbol{\theta}_{DU})+\sigma^2_{D}}\right)\geq \gamma_{D}, \label{transmitpower22222}\\
  & |\boldsymbol{\theta}_{DU}(n)|=1, \forall  n\in\{1,\ldots, N\}. \label{constantmodulus222}
 \end{align}
 \end{subequations}
Here we define $\xi_{U}(\boldsymbol{\theta}_{DU})\triangleq |h_{AU}+\tilde{\mathbf{f}}^H_{AU}\boldsymbol{\theta}_{DU}|^2$, where $\tilde{\mathbf{f}}_{AU}\triangleq \textrm{Diag}(\textrm{conj}(\tilde{\mathbf{f}}_{IU}))\tilde{\mathbf{f}}_{AI}$, and $\xi_{DU}(\boldsymbol{\theta}_{DU})\triangleq |g+\tilde{\mathbf{f}}^H_{DU}\boldsymbol{\theta}_{DU}|^2$, where $\tilde{\mathbf{f}}_{DU}\triangleq \textrm{Diag}(\textrm{conj}(\tilde{\mathbf{f}}_{DI}))\tilde{\mathbf{f}}_{IU}$.

\subsection{HD System}
As for  the corresponding HD system, the received signal  vectors at the AP and DL  user for the three deployment schemes are similar to that of the FD system, but there is no  interference term from the UL user due to the orthogonal UL and DL transmissions. The detailed expressions of the received signals are thus  omitted  for brevity. As a result, the optimization problem of Case 1 in {the HD system} is expressed as
 \begin{subequations} \label{problemhd1}
 \begin{align}
 \min_{p_{U}, p_{A}, \boldsymbol{\theta}_{U}, \boldsymbol{\theta}_{D} }\quad & \kappa_{U} p_{U}+ \kappa_{A} p_{A}  \label{problemhd1objective22}\\
 \mbox{s.t.}\quad
 & \frac{1}{2}\log\left(1+\frac{p_{U}\lambda_{U}(\boldsymbol{\theta}_{U})}{\frac{1}{2}\sigma^2_{A}}\right)\geq \gamma_{A},\label{problemhd1case3:b}\\
 & \frac{1}{2}\log\left(1+\frac{p_{A}\lambda_{D}(\boldsymbol{\theta}_{D})}{\frac{1}{2}\sigma^2_{D}}\right)\geq \gamma_{D}, \label{problemhd1case3:c}\\
  & \eqref{constantmodulus2223}, \notag
 \end{align}
 \end{subequations}
where \eqref{problemhd1case3:b} and \eqref{problemhd1case3:c} denote
the user rate constraints, and  the factor $\frac{1}{2}$ is due to the fact that the uplink and downlink transmissions are allocated with either half of the bandwidth or half of the time as compared to the FD system. The optimization problems for Cases 2 and 3 in the HD system can be similarly formulated, thus they are omitted here.

In the following two  sections, we first analyze  the minimum power consumption required  in different IRS deployment  cases with given IRS passive beamforming vectors. Then, we propose an efficient algorithm to optimize the passive beamforming and thereby solve the above problems.

\section{Performance Analysis}
\label{Section3:analysis}
In this section, we  analyze the minimum transmit power required  in the IRS-aided FD system under each IRS deployment scheme, and compare it with that of the corresponding  HD system. We show that the FD system outperforms its HD counterpart for all IRS deployment schemes, while the distributed deployment further outperforms the other two centralized deployment schemes in the IRS-aided FD system.

\subsection{ Comparison  of FD versus  HD Systems}
\label{subsectionIVA}
In this subsection, we  compare the power consumption of IRS-aided FD system with its HD counterpart under the three considered IRS deployment cases.

\subsubsection{Case 1}
 First, consider the case   of placing  IRSs near both users as shown in Fig. \ref{fig:structure1} (a). For the FD system, since the user rates are monotonically increasing with $p_{U}$ and $p_{A}$,  their  optimal solutions in \eqref{eq:case3} should guarantee that the inequality constraints \eqref{case3:b} and \eqref{case3:c} are met with equality.
Thus, the optimal  transmit power $p_{U}^*$ and $p_{A}^*$ are given by
\begin{equation} \label{optimal_p_1}
p_{U}^*=\frac{(2^{\gamma_{A}}-1)\sigma^2_{A}}{\lambda_{U}(\boldsymbol{\theta}_{U})},
\end{equation}
\begin{equation} \label{optimal_p_2}
p_{A}^*=\frac{(2^{\gamma_{D}}-1)\sigma^2_{A}(2^{\gamma_{A}}-1)\lambda_{DU}(\boldsymbol{\theta}_{U}, \boldsymbol{\theta}_{D} )}{\lambda_{D}(\boldsymbol{\theta}_{D})\lambda_{U}(\boldsymbol{\theta}_{U})}+\frac{(2^{\gamma_{D}}-1)\sigma^2_{D}}{\lambda_{D}(\boldsymbol{\theta}_{D})}.
\end{equation}
By defining $\boldsymbol{\theta}^{*}_{U}$ and $\boldsymbol{\theta}^{*}_{D}$  as the optimal IRS passive beamforming/phase-shift vectors for Case 1, the minimum power consumption of the FD system can be expressed as
\begin{equation}
L_1(\boldsymbol{\theta}^*_{U}, \boldsymbol{\theta}^*_{D})\triangleq  \kappa_{U} p^*_{U}+  \kappa_{A}p^*_{A}= \frac{\bar{\gamma}_1\lambda_{DU}(\boldsymbol{\theta}^*_{U}, \boldsymbol{\theta}^*_{D} )}{\lambda_{D}(\boldsymbol{\theta}^*_{D})\lambda_{U}(\boldsymbol{\theta}^*_{U})}+\frac{\bar{\gamma}_2}{\lambda_{D}(\boldsymbol{\theta}^*_{D})}+\frac{\bar{\gamma}_3}{\lambda_{U}(\boldsymbol{\theta}^*_{U})},\label{eq:31}
\end{equation}
where  $\bar{\gamma}_1\triangleq \kappa_{A}(2^{\gamma_{D}}-1)\sigma^2_{A}(2^{\gamma_{A}}-1)$, $\bar{\gamma}_2\triangleq \kappa_{A}(2^{\gamma_{D}}-1)\sigma^2_{D}$ and $\bar{\gamma}_3\triangleq \kappa_{U}(2^{\gamma_{A}}-1)\sigma^2_{A}$.

By recalling problem \eqref{problemhd1}, similarly, the minimum power consumption of the HD system in Case 1 can be expressed as
\begin{equation}
\tilde{L}_1(\boldsymbol{\theta}^{**}_{U}, \boldsymbol{\theta}^{**}_{D})= \frac{\kappa_{U}(2^{2\gamma_{A}}-1)\sigma^2_{A}}{2\lambda_{U}(\boldsymbol{\theta}^{**}_{U})}+\frac{\kappa_{A}(2^{2\gamma_{D}}-1)\sigma^2_{D}}{2\lambda_{D}(\boldsymbol{\theta}^{**}_{D})},\label{eq:34}
\end{equation}
where $\boldsymbol{\theta}^{**}_{U}\triangleq \arg \max_{\boldsymbol{\theta}_{U}} \lambda_{U}(\boldsymbol{\theta}_{U})$ and $\boldsymbol{\theta}^{**}_{D} \triangleq \arg \max_{\boldsymbol{\theta}_{D}} \lambda_{D}(\boldsymbol{\theta}_{D})$  denote the
channel gain maximization (CGM) based beamforming  vectors at both IRSs, which are optimal  for the HD system.
$\lambda_{U}(\boldsymbol{\theta}^{**}_{U})\triangleq|h_{AU}+\mathbf{f}^H_{AU}\boldsymbol{\theta}^{**}_{U}|^2$ and $\lambda_{D}(\boldsymbol{\theta}^{**}_{D})\triangleq|h_{DA}+\mathbf{g}^H_{DA}\boldsymbol{\theta}^{**}_{D}|^2$ denote the maximum channel gains for the link between the UL user and the AP and that between the AP and the DL user, respectively. It is readily seen that
\begin{equation}
\boldsymbol{\theta}^{**}_{U}=\exp(j(\angle h_{AU} \mathbf{1}-\angle \mathbf{f}_{AU})), \,\, \boldsymbol{\theta}^{**}_{D}=\exp(j(\angle h_{DA} \mathbf{1}-\angle \mathbf{g}_{DA})).\label{eq:suboptimal1111}
\end{equation}
Based on \eqref{eq:31} and \eqref{eq:34}, we have the following theorem.
\begin{theorem}\label{theorem_1}
The minimum  power consumption of the FD system is lower
than that of the HD system in Case 1, i.e., $\tilde{L}_1(\boldsymbol{\theta}^{**}_{U}, \boldsymbol{\theta}^{**}_{D})\geq L_1(\boldsymbol{\theta}^{*}_{U}, \boldsymbol{\theta}^{*}_{D})$, if the following condition is satisfied:
\begin{equation}
\begin{split}
\frac{\kappa_{A}|g|^2}{N^2}\leq U(\boldsymbol{\theta}^{**}_{U}, \boldsymbol{\theta}^{**}_{D}) &\triangleq \frac{\kappa_{U}(|h_{DA}|^2+|\mathbf{g}^H_{DA}\boldsymbol{\theta}^{**}_{D}|^2-2|h_{DA}||\mathbf{g}^H_{DA}\boldsymbol{\theta}^{**}_{D}|)(2^{\gamma_{A}}-1)}{6(2^{\gamma_{D}}-1)N^2}
\\&\quad+\frac{\kappa_{A}(|h_{AU}|^2+|\mathbf{f}^H_{AU}\boldsymbol{\theta}^{**}_{U}|^2-2|h_{AU}||\mathbf{f}^H_{AU}\boldsymbol{\theta}^{**}_{U}|)(2^{\gamma_{D}}-1)\sigma^2_{D}}{6(2^{\gamma_{A}}-1)\sigma^2_{A} N^2}\\&\quad-\frac{\kappa_{A}(|\mathbf{f}^H_{DU}\boldsymbol{\theta}^{**}_{U}|^2+|\mathbf{g}^H_{DU}\boldsymbol{\theta}^{**}_{D}|^2)}{N^2}.\label{conditionanalysis2}
\end{split}
\end{equation}
Moreover, assuming Rayleigh fading channels for all IRS-related links, i.e., $\mathbf{f}_{IU} \sim \mathcal{CN}(\mathbf{0}, \varrho^2_{f_{IU}}\mathbf{I})$, $\mathbf{f}_{AI} \sim \mathcal{CN}(\mathbf{0}, \varrho^2_{f_{AI}}\mathbf{I})$, $\mathbf{g}_{IA} \sim \mathcal{CN}(\mathbf{0}, \varrho^2_{g_{IA}}\mathbf{I})$, $\mathbf{g}_{DI} \sim \mathcal{CN}(\mathbf{0}, \varrho^2_{g_{DI}}\mathbf{I})$, $\mathbf{f}_{DI} \sim \mathcal{CN}(\mathbf{0}, \varrho^2_{f_{DI}}\mathbf{I})$ and $\mathbf{g}_{IU} \sim \mathcal{CN}(\mathbf{0}, \varrho^2_{g_{IU}}\mathbf{I})$,
 it holds that when  $N$ becomes asymptotically large, 
 \begin{equation}
U(\boldsymbol{\theta}^{**}_{U}, \boldsymbol{\theta}^{**}_{D})\rightarrow
\frac{\kappa_{U}(2^{\gamma_{A}}-1)(1-\rho)^2\pi^2\varrho^2_{f_{IU}}\varrho^2_{f_{AI}}}{96(2^{\gamma_{D}}-1)}
+\frac{\bar{\gamma}_2\pi^2\varrho^2_{g_{IA}}\varrho^2_{g_{DI}}\rho^2}{96(2^{\gamma_{A}}-1)\sigma^2_{A}}>0. 
\label{condition22}
\end{equation}
\end{theorem}
\begin{proof}
	Please refer to Appendix \ref{proof_of_The1}.
\end{proof}

Since in  \eqref{conditionanalysis2} $\frac{\kappa_{A}|g|^2}{N^2}\rightarrow 0$ as $N\rightarrow \infty$, the above   condition  is always satisfied for sufficiently large $N$, which implies that with the help of large IRSs, in Case 1 the power consumption of the FD system is asymptotically  lower than that of the HD system, regardless of the interference channel gain $|g|^2$.
Intuitively, this  can be explained by the fact  that with sufficiently large IRSs, the interference between the UL user and the DL user can be effectively suppressed, such that FD is ensured to be more spectral efficient  than HD.

\newtheorem{remark}{Remark}
\begin{remark}
\emph{In contrast,  it can be similarly shown  that without using IRSs,  to guarantee that the power consumption of the FD system is lower than that of the HD system, the following
  condition must be met:
  \begin{equation}
   \kappa_{A}|g|^2\leq \frac{\kappa_{U}(2^{\gamma_{A}}-1)|h_{DA}|^2}{2(2^{\gamma_{D}}-1)}+\frac{\kappa_{A}(2^{\gamma_{D}}-1)|h_{AU}|^2\sigma^2_{D}}{2(2^{\gamma_{A}}-1)\sigma^2_{A}},\label{condition3}
   \end{equation}
which, however,  does not always hold, especially when the UL-to-DL interference is severe, i.e., when $|g|^2$ is large.}
\end{remark}

\subsubsection{Case 2}

Next, we consider  the case of placing all IRS elements to IRS 2  near the DL user. Based on problem \eqref{eq:case1},
 the optimal  transmit power $p_{U}^*$ and $p_{A}^*$ for Case 2 are given by
\begin{equation}
p_{U}^*=\frac{(2^{\gamma_{A}}-1)\sigma^2_{A}}{|h_{AU}|^2},
\end{equation}
\begin{equation}
p_{A}^*=\frac{(2^{\gamma_{D}}-1)\sigma^2_{A}(2^{\gamma_{A}}-1)\omega_{DU}(\boldsymbol{\theta}_{DU})}{|h_{AU}|^2\omega_{D}(\boldsymbol{\theta}_{DU})}+\frac{(2^{\gamma_{D}}-1)\sigma^2_{D}}{\omega_{D}(\boldsymbol{\theta}_{DU})}.
\end{equation}
By defining $\boldsymbol{\theta}^{*}_{DU}$  as the optimal IRS phase-shift vector for Case 2, the minimum power consumption of the FD system can be expressed as
\begin{equation}\label{obj_case2}
L_2(\boldsymbol{\theta}^*_{DU})\triangleq \kappa_{U}p^*_{U}+\kappa_{A}p^*_{A}= \frac{\bar{\gamma}_1\omega_{DU}(\boldsymbol{\theta}^*_{DU})}{|h_{AU}|^2\omega_{D}(\boldsymbol{\theta}^*_{DU})}+\frac{\bar{\gamma}_2}{\omega_{D}(\boldsymbol{\theta}^*_{DU})}
+\frac{\bar{\gamma}_3}{|h_{AU}|^2}.
\end{equation}

Moreover, similar to \eqref{problemhd1case3:b} and \eqref{problemhd1case3:c},  the user rate constraints for the corresponding HD system in this case can be expressed as 
\begin{subequations}\label{problemhd22}
 \begin{align}
  &\frac{1}{2}\log\left(1+\frac{p_{U}|h_{AU}|^2}{\frac{1}{2}\sigma^2_{A}}\right)\geq \gamma_{A},\label{eq:hdcase11}
 \\& \frac{1}{2}\log\left(1+\frac{p_{A}\omega_{D}(\boldsymbol{\theta}_{DU})}{\frac{1}{2}\sigma^2_{D}}\right)\geq \gamma_{D}. \label{eq:hdcase12}
 \end{align}
  \end{subequations}
Then, the minimum power consumption of the HD system in Case 2 can be expressed as
\begin{equation}\label{obj_case2_HD}
\tilde{L}_2(\boldsymbol{\theta}^{**}_{DU})= \frac{\kappa_{U}(2^{2\gamma_{A}}-1)\sigma^2_{A}}{2|h_{AU}|^2}+\frac{\kappa_{A}(2^{2\gamma_{D}}-1)\sigma^2_{D}}{2\omega_{D}(\boldsymbol{\theta}^{**}_{DU})},
\end{equation}
where $\boldsymbol{\theta}^{**}_{DU}\triangleq \arg \max_{\boldsymbol{\theta}_{DU}} \omega_{D}(\boldsymbol{\theta}_{DU})$  denotes the CGM beamforming vector in this case, which is optimal for the  HD system. Similarly, we have
\begin{equation}
\boldsymbol{\theta}^{**}_{DU}=\exp(j(\angle h_{DA} \mathbf{1}-\angle \tilde{\mathbf{g}}_{DA})).\label{eq:suboptima2222}
\end{equation}
The CGM beamforming leads to  the channel gains $\omega_{D}(\boldsymbol{\theta}^{**}_{DU})\triangleq|h_{DA}+\tilde{\mathbf{g}}^H_{DA}\boldsymbol{\theta}^{**}_{DU}|^2$ and  $\omega_{DU}(\boldsymbol{\theta}^{**}_{DU})$ $ \triangleq|g+\tilde{\mathbf{g}}^H_{DU}\boldsymbol{\theta}^{**}_{DU}|^2$ for the link between the AP and the DL user and that between the UL user and the DL user, respectively.
By comparing $L_2(\boldsymbol{\theta}^*_{DU})$ and $\tilde{L}_2(\boldsymbol{\theta}^{**}_{DU})$,  we obtain the following
theorem.
\begin{theorem} \label{theorem_2}
The minimum power consumption of the FD system is lower than that of the HD system in Case 2, if the following condition is satisfied:
\begin{equation}
\begin{split}
\frac{\kappa_{A}|g|^2}{N^2}\leq U(\boldsymbol{\theta}^{**}_{DU}) &\triangleq \frac{\kappa_{U}(2^{\gamma_{A}}-1)(|h_{DA}|^2+|\tilde{\mathbf{g}}^H_{DA}\boldsymbol{\theta}^{**}_{DU}|^2-2|h_{DA}||\tilde{\mathbf{g}}^H_{DA}\boldsymbol{\theta}^{**}_{DU}|)} {4(2^{\gamma_{D}}-1)N^2}\\&\quad+\frac{\kappa_{A}(2^{\gamma_{D}}-1)|h_{AU}|^2\sigma^2_{D}}{4(2^{\gamma_{A}}-1)\sigma^2_{A}N^2}-\frac{\kappa_{A}|\tilde{\mathbf{g}}^H_{DU}\boldsymbol{\theta}^{**}_{DU}|^2}{N^2}.
\label{conditionanalysis2case2}
\end{split}
\end{equation}
Assuming Rayleigh fading channels for all IRS-related links, i.e.,  $\tilde{\mathbf{g}}_{IA} \sim \mathcal{CN}(\mathbf{0}, \varrho^2_{\tilde{g}_{IA}}\mathbf{I})$,  $\tilde{\mathbf{g}}_{DI} \sim \mathcal{CN}(\mathbf{0}, \varrho^2_{\tilde{g}_{DI}}\mathbf{I})$ and $\tilde{\mathbf{g}}_{IU} \sim \mathcal{CN}(\mathbf{0}, \varrho^2_{\tilde{g}_{IU}}\mathbf{I})$, it follows  that when $N$ becomes asymptotically  large, 
 \begin{equation}
U(\boldsymbol{\theta}^{**}_{DU})\rightarrow \frac{\kappa_{U}(2^{\gamma_{A}}-1)(16|h_{DA}|^2+\pi^2\varrho^2_{\tilde{g}_{IA}}\varrho^2_{\tilde{g}_{DI}})}{64(2^{\gamma_{D}}-1)}>0. 
\label{condition2}
\end{equation}
\end{theorem}
\begin{proof}
Please refer to Appendix \ref{appendixA}.
\end{proof}
From the above, we see that with sufficiently large $N$, in Case 2 the FD system always outperforms the HD system, regardless of the interference channel gain $|g|^2$.

 \subsubsection{Case 3}
Last, for the case of placing all reflecting elements to  IRS 1 near the UL user,
the minimum power consumption of the FD system can be expressed as
\begin{equation} L_3(\tilde{\boldsymbol{\theta}}^*_{DU})\triangleq 
\frac{\bar{\gamma}_1\xi_{DU}(\tilde{\boldsymbol{\theta}}^*_{DU})}{|h_{DA}|^2\xi_{U}(\tilde{\boldsymbol{\theta}}^*_{DU})}
	+\frac{\bar{\gamma}_3}{\xi_{U}(\tilde{\boldsymbol{\theta}}^*_{DU})}+\frac{\bar{\gamma}_2}{|h_{DA}|^2}, \label{powerconsumption2222}
\end{equation}
where $\tilde{\boldsymbol{\theta}}^*_{DU}$ denotes the optimal IRS phase-shift vector for Case 3.
By following the same approach for Case 2, we can obtain  the same conclusion for the FD system in Case 3  when $N$ becomes sufficiently large.  Hence, we omit the details for brevity.

\subsection{Comparison of  Different IRS Deployment Cases}
\label{subsectionIVB}

In the following, we focus on the FD system and compare the minimum  power consumption among Case 1, Case 2 and Case 3  to find  which deployment strategy performs the best. In particular, we derive an   asymptotic result   when the number of reflecting elements $N$ becomes large.

With the minimum power consumption in Case 1 expressed in \eqref{eq:31}, we provide the following theorem.
\begin{theorem} \label{Theorem_3}
Assuming Rayleigh fading channels for all IRS-related links, the upper bound of the minimum power consumption of the FD system in Case 1 quadratically decreases to $0$  with the increasing of $N$.
\end{theorem}
\begin{proof}
Please refer to  Appendix \ref{appendixBB}.
\end{proof}
\noindent
Besides, we can   rewrite the expressions of the minimum power consumption of the FD system in \eqref{obj_case2} and \eqref{powerconsumption2222}  for Case 2 and Case 3 as follows:
\begin{equation}
L_2(\boldsymbol{\theta}^*_{DU})=\frac{\bar{\gamma}_1\omega_{DU}(\boldsymbol{\theta}^*_{DU})}{|h_{AU}|^2\omega_{D}(\boldsymbol{\theta}^*_{DU})}+\frac{\bar{\gamma}_2}{\omega_{D}(\boldsymbol{\theta}^*_{DU})}
+\frac{\bar{\gamma}_3}{|h_{AU}|^2}> \frac{\bar{\gamma}_3}{|h_{AU}|^2}>0, \label{powerconsumption1}
\end{equation}
\begin{equation}
	L_3(\tilde{\boldsymbol{\theta}}^*_{DU})=\frac{\bar{\gamma}_1\xi_{DU}(\tilde{\boldsymbol{\theta}}^*_{DU})}{|h_{DA}|^2\xi_{U}(\tilde{\boldsymbol{\theta}}^*_{DU})}
	+\frac{\bar{\gamma}_3}{\xi_{U}(\tilde{\boldsymbol{\theta}}^*_{DU})}+\frac{\bar{\gamma}_2}{|h_{DA}|^2}> \frac{\bar{\gamma}_2}{|h_{DA}|^2}>0.\label{powerconsumption2}
\end{equation}
It is observed from \eqref{powerconsumption1} and \eqref{powerconsumption2} that $L_2(\boldsymbol{\theta}^*_{DU})$ and $L_3(\tilde{\boldsymbol{\theta}}^*_{DU})$ are lower-bounded by $\frac{\bar{\gamma}_3}{|h_{AU}|^2}$ and $\frac{\bar{\gamma}_2}{|h_{DA}|^2}$, respectively. Based on \eqref{powerconsumption1}, \eqref{powerconsumption2} and Theorem 3,
we can conclude that  with sufficiently large $N$, the minimum power consumption in Case 1 is much lower than that in either Case
2 or 3, which means that the distributed deployment generally outperforms the centralized
deployment in terms of power consumption. The reason is that, when $N\rightarrow \infty$, in Case 1 both users have asymptotically large channel gains and the DL user has diminishing interference, whereas Case 2 and Case 3 cannot achieve all the above at the same time.

\section{ passive beamforming  Optimization }
\label{Section4:passivebeamformingalgorithm}
In this section,  we propose a new BCD-based algorithm to optimize the passive beamforming vector at the IRSs, for a general channel model and an arbitrary value of $N$. Subsequently, we discuss the convergence and computational complexity of the proposed algorithm.

\subsection{BCD-Based Passive Beamforming Design}
Let us first focus on Case 1 to introduce the proposed algorithm in general (as Cases 2 and 3 are special cases of Case 1). By substituting the optimal transmit powers in \eqref{optimal_p_1} and \eqref{optimal_p_2} into the original optimization problem \eqref{eq:case3}, we have the following equivalent problem:
\begin{equation}
\begin{split}
\min_{\boldsymbol{\phi}_{U}, \boldsymbol{\phi}_{D}}&\quad \frac{\bar{\gamma}_1|g+\mathbf{f}^H_{DU}\boldsymbol{\psi}_{U}+\mathbf{g}^H_{DU}\boldsymbol{\psi}_{D}|^2 }{|h_{DA}+\mathbf{g}^H_{DA}\boldsymbol{\psi}_{D}|^2|h_{AU}+\mathbf{f}^H_{AU}\boldsymbol{\psi}_{U}|^2}+\frac{\bar{\gamma}_2}{|h_{DA}+\mathbf{g}^H_{DA}\boldsymbol{\psi}_{D}|^2}+\frac{\bar{\gamma}_3}{|h_{AU}+\mathbf{f}^H_{AU}\boldsymbol{\psi}_{U}|^2},
\end{split}\label{problemtheta1}
\end{equation}
where $\boldsymbol{\psi}_{U}\triangleq \exp(j\boldsymbol{\phi}_{U})$, $\boldsymbol{\psi}_{D}\triangleq \exp(j\boldsymbol{\phi}_{D})$, $\boldsymbol{\phi}_{U}\triangleq \angle \boldsymbol{\theta}_{U}$ and $\boldsymbol{\phi}_{D} \triangleq \angle\boldsymbol{\theta}_{D}$. Note that in \eqref{problemtheta1}, the phase-shift values (i.e., $\boldsymbol{\phi}_{U}$ and $\boldsymbol{\phi}_{D}$) of the IRS reflection coefficients are regarded as the optimization variables and $\{\boldsymbol{\psi}_{U},\boldsymbol{\psi}_{D}\}$ are treated as functions of these phase-shift values, thus the original uni-modular constraint \eqref{constantmodulus2223} in problem \eqref{eq:case3} can be safely ignored.

Problem \eqref{problemtheta1} is still  difficult to solve due to the fractional coupling terms in the objective function. In this paper, we first transform \eqref{problemtheta1} into an equivalent but more tractable form and then develop a BCD-based algorithm to solve the equivalent  problem.
By reducing the fractions of the objective function in \eqref{problemtheta1} to a common denominator and taking the reciprocal, we obtain an equivalent expression of problem \eqref{problemtheta1} as
\begin{equation}
\max_{\boldsymbol{\phi}_{U}, \boldsymbol{\phi}_{D}} \quad
\frac{|h_{DA}+\mathbf{g}^H_{DA}\boldsymbol{\psi}_{D}|^2|h_{AU}+\mathbf{f}^H_{AU}\boldsymbol{\psi}_{U}|^2}{\bar{\gamma}_1|g+\mathbf{f}^H_{DU}\boldsymbol{\psi}_{U}+\mathbf{g}^H_{DU}\boldsymbol{\psi}_{D}|^2 + \bar{\gamma}_2|h_{AU}+\mathbf{f}^H_{AU}\boldsymbol{\psi}_{U}|^2 + \bar{\gamma}_3|h_{DA}+\mathbf{g}^H_{DA}\boldsymbol{\psi}_{D}|^2}.
\label{eq_problemtheta2222}
\end{equation}
 Then,  we can  state the following theorem. 
\begin{theorem}
By introducing an auxiliary variable $v\in \mathbb{R}$, problem \eqref{eq_problemtheta2222} can
be equivalently formulated as the following problem
in the sense that both problems share the same global optimal solutions for $\boldsymbol{\phi}_{U}$ and $\boldsymbol{\phi}_{D}$:
\begin{equation}
\begin{split}
\min_{v, \boldsymbol{\phi}_{U}, \boldsymbol{\phi}_{D}}&\quad
v^2(\bar{\gamma}_1|g+\mathbf{f}^H_{DU}\boldsymbol{\psi}_{U}+\mathbf{g}^H_{DU}\boldsymbol{\psi}_{D}|^2 + \bar{\gamma}_2|h_{AU}+\mathbf{f}^H_{AU}\boldsymbol{\psi}_{U}|^2 + \bar{\gamma}_3|h_{DA}+\mathbf{g}^H_{DA}\boldsymbol{\psi}_{D}|^2)\\
&\quad-2\Re(\emph{\textrm{conj}}(v) (h_{DA}+\mathbf{g}^H_{DA}\boldsymbol{\psi}_{D})(h_{AU}+\mathbf{f}^H_{AU}\boldsymbol{\psi}_{U})).
\end{split}
\label{trans_prob}
\end{equation}
\end{theorem}
\begin{proof}
Please refer to Appendix \ref{appendixB}.
\end{proof}

Problem \eqref{trans_prob} can be solved based on the BCD method, where we partition the variables into multiple blocks which are updated sequentially in each iteration.
To this end, the variables are updated as follows: 1) update $v$ by fixing the other variables; 2) update $\boldsymbol{\phi}_{U}(n), \forall n,$ sequentially with the other variables fixed; 3) update $\boldsymbol{\phi}_{D}(n), \forall n,$ sequentially with the other variables fixed.
The detailed updating procedure is presented as follows.

In \textbf{Step 1},  we optimize $v$ by fixing the other variables. The subproblem for $v$ is given by
\begin{equation}
\begin{split}
\min_{v}\quad
v^2a-2\Re(\textrm{conj}(v) b),
\end{split}
\end{equation}
where $a=\bar{\gamma}_1|g+\mathbf{f}^H_{DU}\boldsymbol{\psi}_{U}+\mathbf{g}^H_{DU}\boldsymbol{\psi}_{D}|^2 + \bar{\gamma}_2|h_{AU}+\mathbf{f}^H_{AU}\boldsymbol{\psi}_{U}|^2 + \bar{\gamma}_3|h_{DA}+\mathbf{g}^H_{DA}\boldsymbol{\psi}_{D}|^2$ and $b=(h_{DA}+\mathbf{g}^H_{DA}\boldsymbol{\psi}_{D})(h_{AU}+\mathbf{f}^H_{AU}\boldsymbol{\psi}_{U})$.
We can obtain the optimal solution of this subproblem as $v^{*}=\frac{b}{a}$ by checking the first order optimality condition.

In \textbf{Step 2},  we optimize $\boldsymbol{\phi}_{U}(n), \forall n$, sequentially by fixing the other variables.
The corresponding subproblem for $\boldsymbol{\phi}_{U}(n)$ is given by
\begin{equation}
\begin{split}
\min_{\boldsymbol{\phi}_{U}(n)}&\quad
v^2(\bar{\gamma}_1|g+\mathbf{f}^H_{DU}\boldsymbol{\psi}_{U}+\mathbf{g}^H_{DU}\boldsymbol{\psi}_{D}|^2 + \bar{\gamma}_2|h_{AU}+\mathbf{f}^H_{AU}\boldsymbol{\psi}_{U}|^2 )\\
&\quad-2\Re(\textrm{conj}(v)(h_{DA}+\mathbf{g}^H_{DA}\boldsymbol{\psi}_{D})(h_{AU}+\mathbf{f}^H_{AU}\boldsymbol{\psi}_{U})).
\end{split}\label{phi1_subproblem}
\end{equation}

By appropriately rearranging the objective function of problem \eqref{phi1_subproblem}, we have the following equivalent problem:
\begin{equation}
\begin{split}
\min_{\boldsymbol{\phi}_{U}(n)}&\quad
\boldsymbol{\psi}_{U}^H\mathbf{A}_{\psi_{U}}\boldsymbol{\psi}_{U} - 2\Re(\boldsymbol{\psi}_{U}^H\mathbf{b}_{\psi_{U}}),
\end{split}\label{phi1_subproblem_trans}
\end{equation}
where $\mathbf{A}_{\psi_{U}}\triangleq |v|^2(\bar{\gamma}_1\mathbf{f}_{DU}\mathbf{f}_{DU}^H + \bar{\gamma}_2\mathbf{f}_{AU}\mathbf{f}_{AU}^H)$ and $\mathbf{b}_{\psi_{U}}\triangleq v( \textrm{conj}(h_{DA})+\boldsymbol{\phi}_{D}^H\mathbf{g}_{DA})\mathbf{f}_{AU} - |v|^2( \bar{\gamma}_1\mathbf{f}_{DU}(g+\mathbf{g}_{DU}^H\boldsymbol{\phi}_{D}) + \bar{\gamma}_2\mathbf{f}_{AU}h_{AU} )$.
It is readily seen that the objective function of problem \eqref{phi1_subproblem_trans} is a quadratic function with respect to $\boldsymbol{\psi}_{U}(n)$. Hence, by omitting the irrelevant constant terms, we can rewrite problem \eqref{phi1_subproblem_trans} as follows:
\begin{equation}
\begin{split}
\min_{\boldsymbol{\phi}_{U}(n)}&\quad
\bar{a}_{\psi_{U,n}}|\boldsymbol{\psi}_{U}(n)|^2 - 2\Re(\textrm{conj}(\bar{b}_{\psi_{U,n}})\boldsymbol{\psi}_{U}(n)),
\end{split}\label{phi1_subproblem_final}
\end{equation}
where $\bar{a}_{\psi_{U,n}}$ is a known coefficient that is not related to $\boldsymbol{\phi}_{U}(n)$, $\bar{b}_{\psi_{U,n}} = \mathbf{A}_{\psi_{U}}(n, n)\boldsymbol{\psi}_{U}(n) - \mathbf{A}_{\psi_{U}}(n, :)\boldsymbol{\psi}_{U} + \mathbf{b}_{\psi_{U}}(n),$ and $\mathbf{A}_{\psi_{U}}(n, :)$ denotes the $n$th row vector of matrix $\mathbf{A}_{\psi_{U}}$. Obviously, the optimal solution of problem \eqref{phi1_subproblem_final} is given by $\boldsymbol{\phi}^{*}_{U}(n)=\angle \bar{b}_{\psi_{U,n}}$.

In \textbf{Step 3},  we optimize $\boldsymbol{\phi}_{D}(n), \forall n$, sequentially by fixing the other variables. The corresponding subproblem is given by
\begin{equation}
\begin{split}
\min_{\boldsymbol{\phi}_{D}(n)}&\quad
v^2(\bar{\gamma}_1|g+\mathbf{f}^H_{DU}\boldsymbol{\psi}_{U}+\mathbf{g}^H_{DU}\boldsymbol{\psi}_{D}|^2 + \bar{\gamma}_3|h_{DA}+\mathbf{g}^H_{DA}\boldsymbol{\psi}_{D}|^2 )\\
&\quad-2\Re(\textrm{conj}(v(h_{DA})+\mathbf{g}^H_{DA}\boldsymbol{\psi}_{D})(h_{AU}+\mathbf{f}^H_{AU}\boldsymbol{\psi}_{U})).
\end{split}\label{phi2_subproblem}
\end{equation}
This subproblem can be similarly solved as problem \eqref{phi1_subproblem}, thus the details are omitted for brevity.

In summary, we can solve problem \eqref{problemtheta1} by iterating over the abovementioned three steps and the overall  BCD-based
algorithm is summarized  in Algorithm \ref{tab:table1}.
\begin{algorithm}
  \centering
  \caption{ Proposed BCD-based passive beamforming design algorithm}
 \begin{itemize}
 \item[1.] Define the tolerance of accuracy $\delta$. Initialize the algorithm with a feasible point. Set the iteration number $i=0$ and the maximum iteration number $N_{\max }$.
 \item[2.] \textbf{Repeat}
    \begin{itemize}
    \item Update $v$, $\boldsymbol{\phi}_{U}(n), \forall n$ and $\boldsymbol{\phi}_{D}(n), \forall n$ sequentially according to \textbf{Steps} 1-3, respectively.
   \item Update the iteration number : $i \leftarrow i+1$.
  \end{itemize}
 \item[3.] \textbf{Until}
the fractional decrease of \eqref{trans_prob} is less than $\delta$ or the maximum number of iterations is reached, i.e., $i > {N_{\max }}$.
 \end{itemize}
\label{tab:table1}
\end{algorithm}

Besides, the passive beamforming optimization problems for Case 2 and Case 3 can be written as
\begin{equation}
\begin{split}
\min_{\boldsymbol{\phi}_{DU}} & \quad \frac{\bar{\gamma}_1|g+\tilde{\mathbf{g}}^H_{DU}\boldsymbol{\psi}_{DU}|^2}{|h_{AU}|^2|h_{DA}+\tilde{\mathbf{g}}^H_{DA}\boldsymbol{\psi}_{DU}|^2}+\frac{\bar{\gamma}_2}{|h_{DA}+\tilde{\mathbf{g}}^H_{DA}\boldsymbol{\psi}_{DU}|^2},
\end{split}\label{problemtheta12222}
\end{equation}
\begin{equation}
\begin{split}
\min_{\tilde{\boldsymbol{\phi}}_{DU}}&\quad \frac{\bar{\gamma}_1|g+\tilde{\mathbf{f}}^H_{DU}\tilde{\boldsymbol{\psi}}_{DU}|^2}{|h_{DA}|^2|h_{AU}+\tilde{\mathbf{f}}^H_{AU}\tilde{\boldsymbol{\psi}}_{DU}|^2}
+\frac{\bar{\gamma}_3}{|h_{AU}+\tilde{\mathbf{f}}^H_{AU}\tilde{\boldsymbol{\psi}}_{DU}|^2},
\end{split}\label{problemtheta13333}
\end{equation}
respectively, where $\boldsymbol{\psi}_{DU}\triangleq \exp(j\boldsymbol{\phi}_{DU})$, $\boldsymbol{\phi}_{DU} \triangleq \angle\boldsymbol{\theta}_{DU}$, $\tilde{\boldsymbol{\psi}}_{DU}\triangleq \exp(j\tilde{\boldsymbol{\phi}}_{DU})$ and $\tilde{\boldsymbol{\phi}}_{DU} \triangleq \angle\tilde{\boldsymbol{\theta}}_{DU}$.
Since Algorithm \ref{tab:table1} can be easily modified  to solve problems \eqref{problemtheta12222} and \eqref{problemtheta13333} as well, the details are omitted  here.

\subsection{Complexity and Convergence of Algorithm 1}
The complexity of  Algorithm \ref{tab:table1} is dominated by the matrix multiplication operations required for updating $\boldsymbol{\phi}_{U}$ and $\boldsymbol{\phi}_{D}$. Thus, by omitting the lower order terms, the overall complexity of Algorithm \ref{tab:table1} is given by $\mathcal{O}(IN^2)$, where $I$ is the iteration number required by Algorithm \ref{tab:table1}.

As for the convergence property of Algorithm \ref{tab:table1}, it is  seen that each subproblem in the proposed BCD-based algorithm is globally and uniquely solved. Therefore, according to Proposition 2.7.1 (Convergence of Block Coordinate Descent) in  \cite{Bertsekas1999}, Algorithm \ref{tab:table1} is guaranteed to converge to a  Karush--Kuhn--Tucker  (KKT) point of problem \eqref{trans_prob}, i.e., a solution satisfying the KKT conditions of the problem. Furthermore, according to \cite{Shen2018TSP}, problem \eqref{trans_prob} and problem \eqref{eq_problemtheta2222} share the same KKT points. Therefore, Algorithm \ref{tab:table1} is guaranteed to converge to a KKT point of problem \eqref{eq_problemtheta2222}, which is also a KKT point of problem \eqref{problemtheta1}.

\section{ Simulation results}
\label{Section5:simulations}

In this section, we validate our  analysis and evaluate the performance of the proposed algorithm based on simulations.
The path loss is modeled as $L = C_0(d_{\textrm{link}}/D_0)^{-\alpha}$, where $C_0$ is the path loss at the {reference} distance of $D_0=1$ m, $\alpha$ is the path-loss exponent and $d_{\textrm{link}}$ represents the link distance. We denote the path-loss exponents of the AP-user (both UL and DL), AP-IRS, IRS-user  and user-user links as $\alpha_{DA}$, $\alpha_{IA}$ $\alpha_{DI}$ and $\alpha_{DU}$, respectively, and  set $\alpha_{DA}=3.6$ and $\alpha_{IA}=\alpha_{DI}=\alpha_{DU}=2.2$. In the simulations, we construct a 3-D coordinate system where the AP is located on the $x$-axis and the IRSs are placed in the planes parallel to the $y-z$ plane, as illustrated in Fig. \ref{user_setup}. The distance between the UL user and the DL user is represented by $d$. The reference antenna at AP is located at $(2+\frac{d}{2}\;\textrm{m}, 0, 0)$, the UL user is located at $(2\;\textrm{m}, 150\;\textrm{m}, 0)$, and the DL user is located at $(d+2\;\textrm{m}, 150\;\textrm{m}, 0)$. For Case 1, the reference reflecting elements at the two  IRSs are placed at $(0, 150\;\textrm{m}, 3\;\textrm{m})$ and  $(d+4 \;\textrm{m}, 150\;\textrm{m}, 3\;\textrm{m})$, respectively. For Case 2 and Case 3, the reference reflecting element at the combined IRS is placed at $(d+4\;\textrm{m}, 150\;\textrm{m}, 3\;\textrm{m})$ and $(0, 150\;\textrm{m}, 3\;\textrm{m})$, respectively.
\begin{figure*}[!t]
\centering
	\setlength{\abovecaptionskip}{-0.2cm}
\setlength{\belowcaptionskip}{-0.2cm}
\scalebox{0.47}{\includegraphics{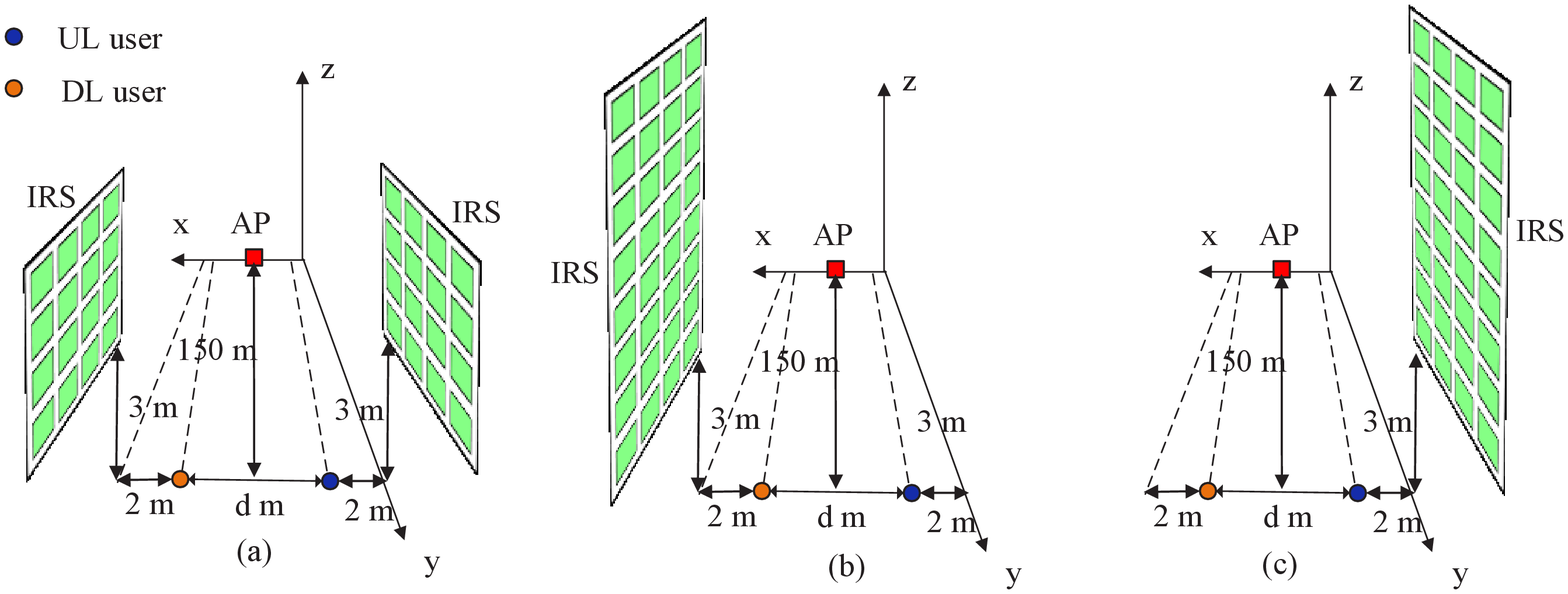}}
\caption{Simulation setup. (a) Case 1: Place IRSs on both user sides; (b) Case 2: IRS is placed  near the DL user only; (c) Case 3: IRS is placed  near the UL user only. }\label{user_setup}
\end{figure*}
Besides, we assume the Rician fading channel model for all links in our simulations since both line-of-sight (LoS) and non-LoS (NLoS) components may exist in practical channels. Accordingly, the AP-IRS  and  IRS-user channels are modeled as
\begin{equation}
	\mathbf{f}=\sqrt{\frac{\beta}{1+\beta}}\mathbf{f}^{\textrm{LoS}}+\sqrt{\frac{1}{1+\beta}}\mathbf{f}^{\textrm{NLoS}},
\end{equation}
where $\beta$ is the Rician factor, $\mathbf{f}^{\textrm{LoS}}$ and $\mathbf{f}^{\textrm{NLoS}}$ denote the deterministic LoS and NLoS components, respectively.
Here we denote the Rician factors of the AP-IRS  and IRS-user links as $\beta_{IA}$ and $\beta_{DI}$, respectively, and we set $\beta_{IA}=9$ dB and $\beta_{DI}=6$ dB in our simulations. $g$ is modeled similarly, with the Rician factor equal to  $4$ dB.  Without loss of generality, we set $\rho=0.5$ and $\kappa_U=\kappa_A=1$ in all of our simulations.

In the following, we first show  the advantage of using IRS over without using IRS by comparing the minimum transmit power required  in the FD and HD systems  under  different IRS  deployment strategies. Then, we compare the performance of the three IRS deployment strategies based on the FD system.

\subsection{FD versus HD Systems}
In this subsection, we aim to verify our analytical results  in Section \ref{subsectionIVA} by simulations. Thus, the IRS reflection coefficients for the considered three deployment strategies are set based on the CGM beamforming given  in Section \ref{Section3:analysis} to maximize the channel power  gains of the IRS-related links. Unless otherwise specified, we set $d = 80$ m and $\gamma_{A}=\gamma_{D}=4$ bits/s/Hz.

\begin{figure}[!t]
\centering
\scalebox{0.56}{\includegraphics{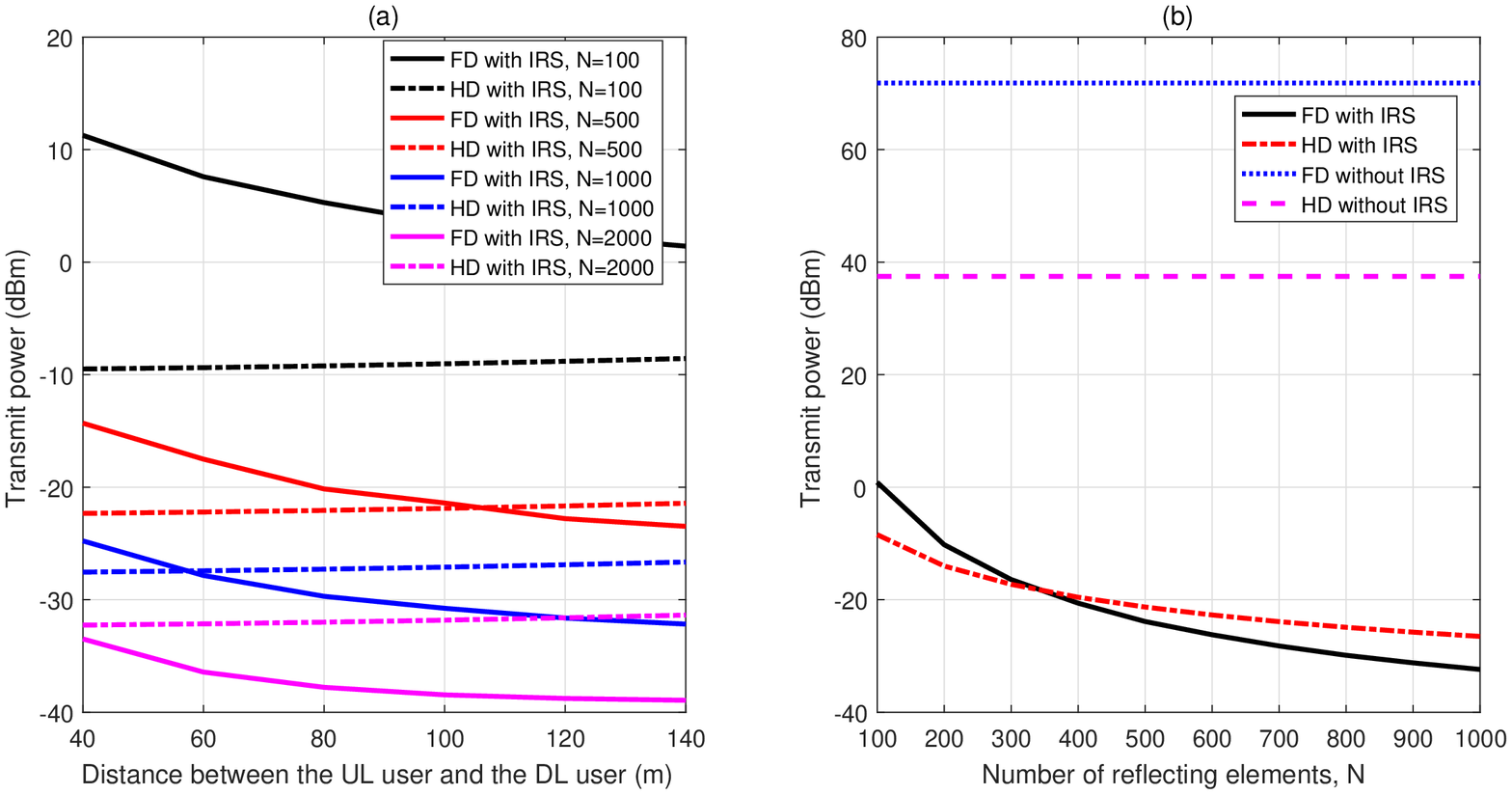}}
\caption{Minimum power consumption in Case 1: (a) Transmit power versus the distance between the UL and DL users; (b) Transmit power versus the total  number of IRS  reflecting elements.}\label{Fig1}
\end{figure}

First, consider the performance  comparison between the FD and HD systems in Case 1. 
Fig. \ref{Fig1} shows the minimum total transmit power consumption (for both the AP and UL user)  versus the distance between the UL and DL users, $d$, as well as   the total  number of IRS reflecting elements, $N$. From Fig. \ref{Fig1} (a), we can see that the power consumption of the FD system decreases with the  increasing of $d$ while the power consumption of the HD system gradually increases with the  increasing of $d$. This is because the UL-to-DL interference in general weakens as $d$ increases, which is beneficial for improving the performance of the FD system; on the other hand, larger $d$ implies longer distance between the AP and users and more severe path losses of the AP-user links, therefore although there is no UL-to-DL interference in the HD system, its performance could be impaired as $d$ increases. 
 Moreover, it is observed that the transmit powers of all considered systems decrease with the increasing of $N$, which is expected since larger $N$ leads to higher passive beamforming gain.
Fig. \ref{Fig1} (b) illustrates the relationship between the transmit power and the number of reflecting elements at the IRSs for both FD and HD systems. 
 From the results, we can see that the transmit power of the FD system is larger than that of the HD system when there is no IRS or $N$ is small; however the transmit power of the FD system becomes much smaller than that of the HD system  with the increasing of $N$, which verifies the results in Theorem \ref{theorem_1}. 

\begin{figure}[!t]
\centering
\scalebox{0.56}{\includegraphics{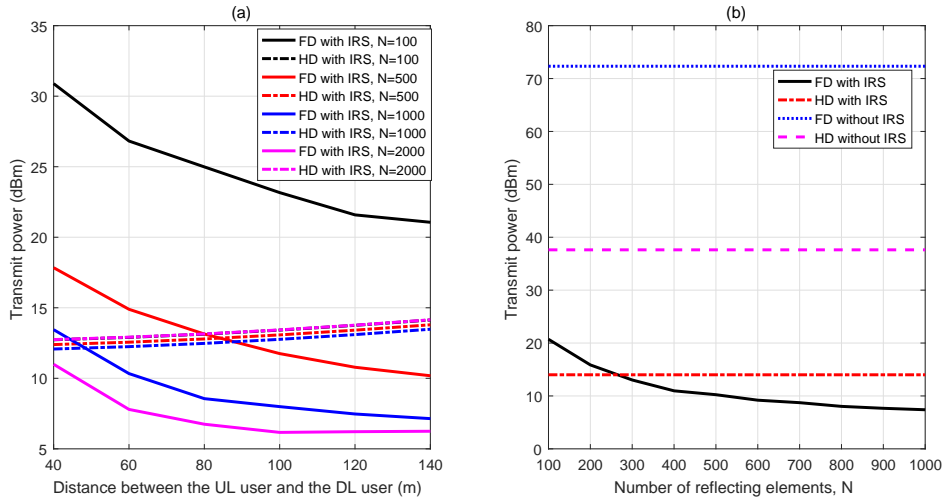}}
\caption{Minimum power consumption in Case 2: (a) Transmit power versus the distance between the UL  and  DL users; (b) Transmit power versus the total number of IRS reflecting elements.}\label{Fig2}
\end{figure}

\begin{figure}[!t]
\centering
\scalebox{0.56}{\includegraphics{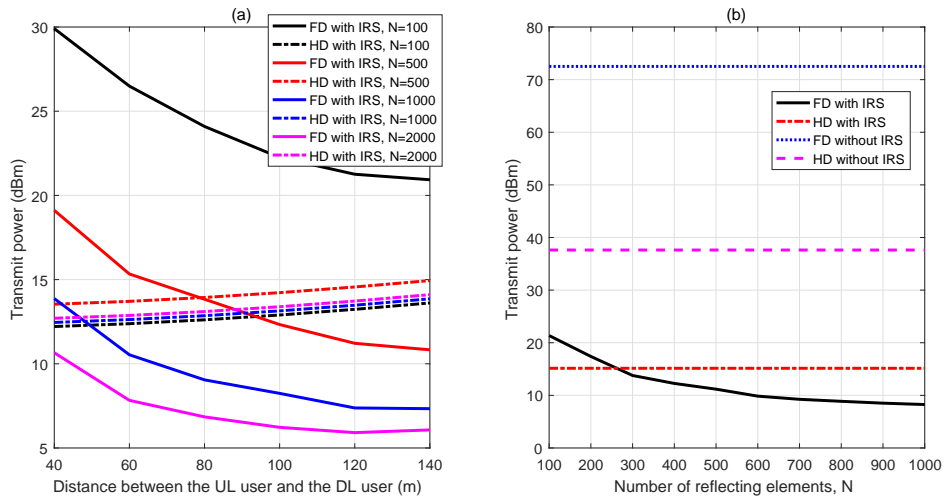}}
\caption{Minimum power consumption in Case 3: (a) Transmit power versus the distance between the UL  and  DL users; (b) Transmit power versus the total number of IRS reflecting elements.}\label{Fig3}
\end{figure}

Then, we compare the transmit power consumption in the  FD and HD systems for Case 2 and Case 3.
Fig. \ref{Fig2} (a) and   (b) illustrate the transmit power performance versus the distance $d$ and the number of IRS reflecting elements  in Case 2, respectively, and
the same comparison is shown in Fig. \ref{Fig3} (a) and  (b) for Case 3. 
Similar to Case 1, from Fig. \ref{Fig2} (a)  and Fig. \ref{Fig3} (a), the transmit power of the FD system decreases while that of the HD system slightly increases with $d$ in both Case 2 and Case 3.
The reason for this phenomenon is the same as the aforementioned one for Case 1. 
In Fig. \ref{Fig2} (b)  and Fig. \ref{Fig3} (b), the trend of the curves regarding the HD  system with IRS  does not decrease much 
 with the increasing  number of reflecting elements,  which is  different from its counterpart in Case 1. 
This  can be explained by the fact that  with only one large IRS deployed near the DL (or UL) user, the power consumption required for achieving the rate requirement of the UL  (DL) user cannot be reduced, 
thus the total transmit power consumption  in Case 2 (Case 3) is dominated by the power  required by the UL (DL) transmission  which remains unchanged under different values of $N$. 
 From the above  results,  we  can see  that  the FD and HD systems with IRS achieve much lower transmit power consumption  as compared to those without IRS. Furthermore, the transmit power of the FD system becomes much smaller than that of the HD system with the increasing of $N$.
The  results thus verify the discussions in Section \ref{subsectionIVA} for Cases 2 and  3.

\subsection{Comparison of Different IRS Deployment Strategies  in FD System}

In this subsection, we evaluate the performance of our proposed BCD-based passive beamforming algorithm (i.e., Algorithm \ref{tab:table1}) in the FD system and compare the performance of different IRS deployment strategies. Unless otherwise specified, we assume $\gamma_{A}=\gamma_{D}=4$ bits/s/Hz, $N=2000$ and $d=40$ m for all cases.

\begin{figure}[!t]
\centering
\scalebox{0.56}{\includegraphics{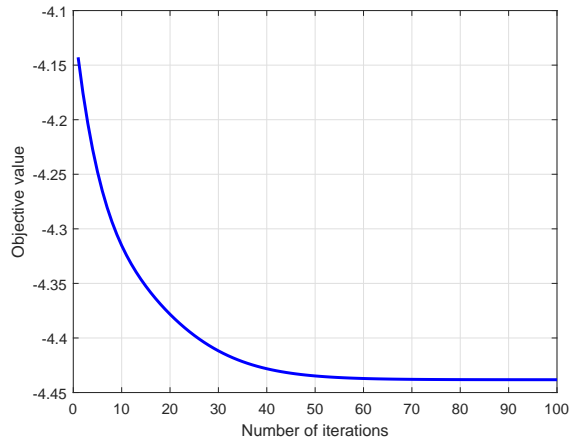}}
\caption{Convergence  behavior of the proposed BCD-based passive beamforming algorithm.}\label{Fig4}
\end{figure}

We first investigate the convergence behavior of our proposed BCD-based algorithm. Fig. \ref{Fig4} shows the achieved objective value versus the number of iterations in Case 1.
The phase shifts are randomly initialized within $[0,2\pi)$. It is observed that the transmit power achieved by the proposed algorithm decreases monotonically with the number of iterations and the proposed algorithm can converge  within 60 iterations. Moreover, the proposed algorithm has a  low complexity due to the closed-form solution (multiplication operation only) of each step shown in Algorithm \ref{tab:table1}.

Next, we study in Fig. \ref{Fig5} the effect of the rate requirement of the UL user, i.e.,  $\gamma_{A}$, on the transmit power consumption, under the  three IRS deployment strategies. 
We can see from this figure that the transmit power  of Case 1 is significantly lower than that of Case 2 or Case 3, which shows the advantage of distributed deployment over centralized deployment under the considered FD system. Besides, we observe that the transmit power in Case 3 increases slower than that in Case 2 with the increase of $\gamma_{A}$. This is because in Case 3, the IRS is placed near the UL user and thus the rate requirement of the UL user can be met easier.
Moreover, we can see that the proposed BCD-based passive beamforming algorithm achieves better performance compared to the CGM passive beamforming algorithm, especially when $\gamma_A$ is large. This is mainly because the UL-to-DL interference becomes the performance bottleneck in the high-$\gamma_A$ regime and simply employing the CGM based beamforming cannot effectively suppress this interference.
\begin{figure}[!t]
\centering
\scalebox{0.56}{\includegraphics{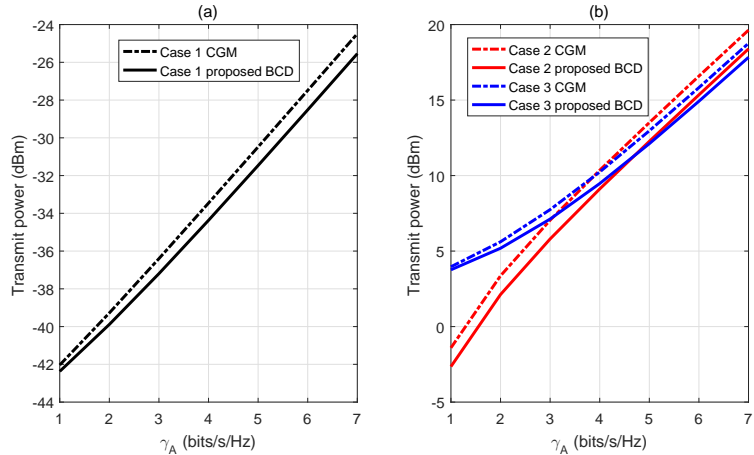}}
\caption{ Transmit power versus $\gamma_{A}$ for the proposed BCD-based algorithm and the CGM algorithm: (a) Case 1; (b) Case 2 and Case 3.}\label{Fig5}
\end{figure}
Similarly, we explore in Fig. \ref{Fig6} the effect of the rate requirement of the DL user, i.e.,  $\gamma_{D}$, on the transmit power consumption. We can see that Case 1 also outperforms Case 2 and Case 3 under  this setup. Besides, the transmit power in Case 2 increases slower than that in Case 3 as $\gamma_{D}$ increases since the rate requirement of the DL user can be more easily satisfied in this case.

\begin{figure}[!t]
\centering
\scalebox{0.56}{\includegraphics{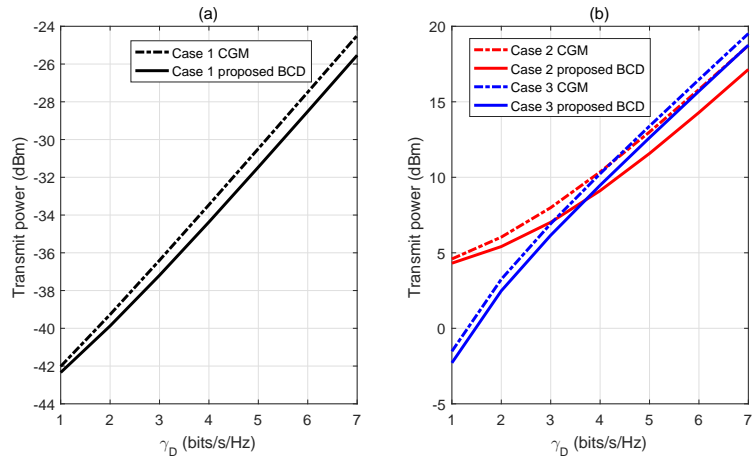}}
\caption{Transmit power versus $\gamma_{D}$ for the proposed BCD-based algorithm and the CGM algorithm: (a) Case 1; (b) Case 2 and Case 3.}\label{Fig6}
\end{figure}

\begin{figure}[!t]
\centering
\scalebox{0.56}{\includegraphics{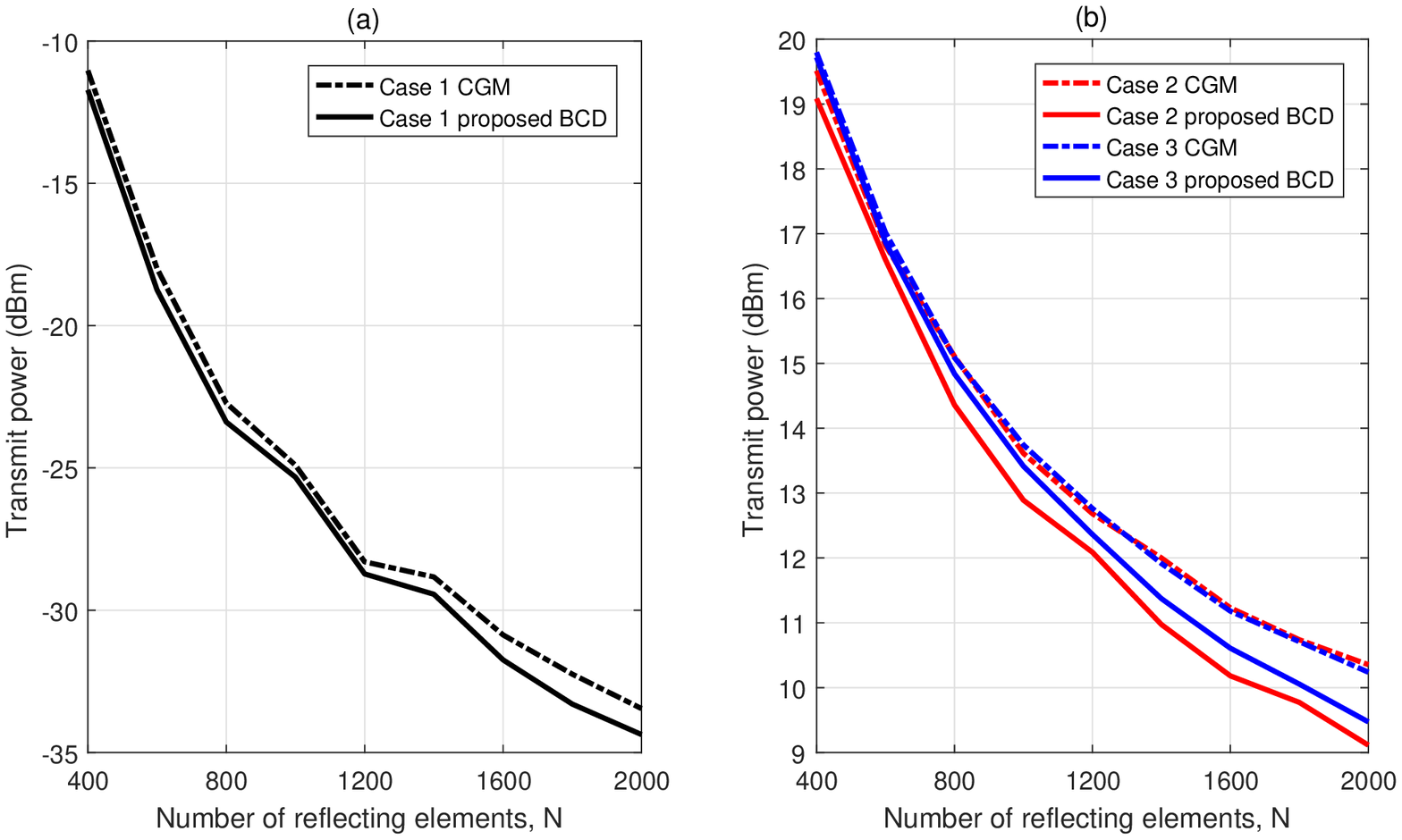}}
\caption{Transmit power versus the number of IRS reflecting elements  for the proposed BCD-based algorithm and the CGM algorithm: (a) Case 1; (b) Case 2 and Case 3.}\label{Fig7}
\end{figure}

Finally, in Fig. \ref{Fig7}, we investigate the effect of the number of IRS reflecting elements  on the transmit power  performance for the three IRS deployment strategies. It can be seen that the performance of Case 1 is the best among all considered. Besides, the distributed deployment strategy in Case 1 also provides the largest performance improvement when $N$ is increased from $400$ to $2000$, e.g., the transmit power when $N=2000$ is about  $20$ dBm lower than that when $N=400$ in Case 1, but this gain  is only around  $10$ dBm in Cases 2 and 3. These results coincide with the analysis in  Section \ref{subsectionIVB}.
Furthermore, we can see that the proposed BCD-based algorithm still outperforms the CGM algorithm and the gain increases with the number of reflecting elements. This is because larger $N$ offers more flexibility when designing the reflection coefficients and a more significant  tradeoff between channel gain and UL-to-DL interference can be achieved by the proposed BCD-based algorithm, while only channel gain maximization is considered in the CGM algorithm.

\section{Conclusions and Future Work} 
\label{Section6:conclusion}
In this paper, we have studied three deployment  cases for an IRS-aided FD wireless system. In each case, the weighted sum transmit power consumption of the AP and UL user is minimized  by jointly optimizing their transmit power  and the passive reflection coefficients at the IRS (or IRSs), subject to the UL and DL users' rate constraints and the uni-modulus constraints on the IRS reflection coefficients.
We have analyzed the minimum transmit power in the IRS-aided FD system under the three deployment schemes, as compared with that of the corresponding  HD system. 
Specifically, we have showed that the FD system outperforms its HD counterpart for all IRS deployment schemes, while the distributed deployment further outperforms the other two centralized deployment schemes. 
Moreover,  we have developed an efficient BCD-based algorithm  for passive beamforming design in a general system setup. 
Finally, numerical results have been presented to verify our analysis and the efficacy of the  proposed passive beamforming design algorithm.

In the following, we briefly discuss some important issues/aspects of this work that are not addressed yet, so as to motivate  future research.
\begin{itemize}
	\item For simplicity, we only considered the single UL/DL user case and the two users are assumed to be equipped with a single antenna each. Besides, the FD AP is also assumed to have a single transmit/receive antenna. Generally, there may exist multiple UL/DL users, and the AP as well as users are usually  equipped with multiple antennas for high-rate  communications, therefore the joint design of active beamforming at the AP/users and passive beamforming at the IRS is an important problem to investigate in the future.
\item In this paper, we considered continuous phase shifts at the reflecting elements of the IRS, which may be practically difficult to implement due to the hardware limitation. Besides, the reflection amplitudes at the IRS were assumed to be 1 for simplicity and the potential  joint optimization of  reflection amplitudes and phase shifts  \cite{Zhao2020intelligent} was not exploited. Thus, beamforming optimization with discrete phase shifts at the IRS and joint reflection amplitude and phase-shift optimization are interesting topics, which are worthy of further investigation.
\item In practice, the SI due to the FD operation cannot be perfectly canceled in general and CSI errors are inevitable in practice due to limited channel training resources. Therefore, beamforming optimization under more practice channel model is also an important problem to investigate in future work.
\end{itemize}
\begin{appendices}

\section{Proof for Theorem \ref{theorem_1}} \label{proof_of_The1}
By comparing $L_1(\boldsymbol{\theta}^{*}_{U}, \boldsymbol{\theta}^{*}_{D})$ and $\tilde{L}_1(\boldsymbol{\theta}^{**}_{U}, \boldsymbol{\theta}^{**}_{D})$,
we obtain
\begin{equation}
	\begin{split}
		&\tilde{L}_1(\boldsymbol{\theta}^{**}_{U}, \boldsymbol{\theta}^{**}_{D})-L_1(\boldsymbol{\theta}^{*}_{U}, \boldsymbol{\theta}^{*}_{D}) \\& =
		\frac{\kappa_{U}(2^{\gamma_{A}}+1)(2^{\gamma_{A}}-1)\sigma^2_{A}}{2\lambda_{U}(\boldsymbol{\theta}^{**}_{U})}+\frac{\kappa_{A}(2^{\gamma_{D}}+1)(2^{\gamma_{D}}-1)\sigma^2_{D}}{2\lambda_{D}(\boldsymbol{\theta}^{**}_{D})}
		\\&\quad
		-\left( \frac{\bar{\gamma}_3}{\lambda_{U}(\boldsymbol{\theta}^{*}_{U})}+\frac{\bar{\gamma}_2}{\lambda_{D}(\boldsymbol{\theta}^{*}_{D})}+
		\frac{\bar{\gamma}_1\lambda_{DU}(\boldsymbol{\theta}^*_{U}, \boldsymbol{\theta}^*_{D} )}{\lambda_D(\boldsymbol{\theta}^*_{D})\lambda_{U}(\boldsymbol{\theta}^*_{U})}\right)
		\\&  \geq 
		\tilde{L}_1(\boldsymbol{\theta}^{**}_{U}, \boldsymbol{\theta}^{**}_{D})-L_1(\boldsymbol{\theta}^{**}_{U}, \boldsymbol{\theta}^{**}_{D})
		\\&=\frac{\kappa_{U}(2^{\gamma_{A}}-1)^2\sigma^2_{A}}{2\lambda_{U}(\boldsymbol{\theta}^{**}_{U})}+\frac{\kappa_{A}(2^{\gamma_{D}}-1)^2\sigma^2_{D}}{2\lambda_{D}(\boldsymbol{\theta}^{**}_{D})}-
		\frac{\bar{\gamma}_1\lambda_{DU}(\boldsymbol{\theta}^{**}_{U}, \boldsymbol{\theta}^{**}_{D} )}{\lambda_{D}(\boldsymbol{\theta}^{**}_{D})\lambda_{U}(\boldsymbol{\theta}^{**}_{U})}, 
	\end{split}
\end{equation}
where $\lambda_{DU}(\boldsymbol{\theta}^{**}_{U}, \boldsymbol{\theta}^{**}_{D} ) \triangleq |g+ \mathbf{f}^H_{DU}\boldsymbol{\theta}^{**}_{U}+\mathbf{g}^H_{DU}\boldsymbol{\theta}^{**}_{D}|^2$ denotes the channel power  gain between the UL user and the DL user
with  the CGM beamforming and
the inequality  holds since the CGM beamforming vectors $\boldsymbol{\theta}^{**}_{U}$ and $\boldsymbol{\theta}^{**}_{D}$ shown in \eqref{eq:suboptimal1111}  are suboptimal solutions of problem \eqref{eq:case3} in general and
$L_1(\boldsymbol{\theta}^{**}_{U}, \boldsymbol{\theta}^{**}_{D})\geq L_1(\boldsymbol{\theta}^{*}_{U}, \boldsymbol{\theta}^{*}_{D})$. 
Note that in light of the Cauchy-Schwarz inequality\footnote{Our result also holds  by the fact that for two complex
	numbers $x$ and $y$, the following inequality holds $||x|-|y||\leq |x+y|\leq |x|+|y|$.}, we can obtain
\begin{equation}
	\begin{split}
		\lambda_{DU}(\boldsymbol{\theta}^{**}_{U}, \boldsymbol{\theta}^{**}_{D} )
		\leq (|g|+|\mathbf{f}^H_{DU}\boldsymbol{\theta}^{**}_{U}|+|\mathbf{g}^H_{DU}\boldsymbol{\theta}^{**}_{D}|)^2\leq 3(|g|^2+|\mathbf{f}^H_{DU}\boldsymbol{\theta}^{**}_{U}|^2+|\mathbf{g}^H_{DU}\boldsymbol{\theta}^{**}_{D}|^2).
	\end{split}
\end{equation}
Thus,  we further have
\begin{equation}
	\begin{split}
		&\tilde{L}_1(\boldsymbol{\theta}^{**}_{U}, \boldsymbol{\theta}^{**}_{D})-L_1(\boldsymbol{\theta}^{*}_{U}, \boldsymbol{\theta}^{*}_{D})
		\\& \geq 
		\frac{\kappa_{U}(2^{\gamma_{A}}-1)^2\sigma^2_{A}}{2\lambda_{U}(\boldsymbol{\theta}^{**}_{U})}+\frac{\kappa_{A}(2^{\gamma_{D}}-1)^2\sigma^2_{D}}{2\lambda_{D}(\boldsymbol{\theta}^{**}_{D})}-
		\frac{3\bar{\gamma}_1(|g|^2+|\mathbf{f}^H_{DU}\boldsymbol{\theta}^{**}_{U}|^2+|\mathbf{g}^H_{DU}\boldsymbol{\theta}^{**}_{D}|^2)         }{\lambda_{D}(\boldsymbol{\theta}^{**}_D)\lambda_{U}(\boldsymbol{\theta}^{**}_{U})}.\label{eq:differencecase3appendix}
	\end{split}
\end{equation}
By multiplying $\frac{\lambda_{D}(\boldsymbol{\theta}^{**}_{D})\lambda_{U}(\boldsymbol{\theta}^{**}_{U})}{3(2^{\gamma_{D}}-1)\sigma^2_{A}(2^{\gamma_{A}}-1)N^2}>0$ on both sides of \eqref{eq:differencecase3appendix}, it yields 
\begin{equation}
	\begin{split}
		&\frac{\lambda_{D}(\boldsymbol{\theta}^{**}_{D})\lambda_{U}(\boldsymbol{\theta}^{**}_{U})(\tilde{L}_1(\boldsymbol{\theta}^{**}_{U}, \boldsymbol{\theta}^{**}_{D})-L_1(\boldsymbol{\theta}^{*}_{U}, \boldsymbol{\theta}^{*}_{D}))}{3(2^{\gamma_{D}}-1)\sigma^2_{A}(2^{\gamma_{A}}-1)N^2} \\& \geq
		\frac{\kappa_{U}\lambda_{D}(\boldsymbol{\theta}^{**}_{D})(2^{\gamma_{A}}-1)}{6(2^{\gamma_{D}}-1)N^2}+\frac{\kappa_{A}\lambda_{U}(\boldsymbol{\theta}^{**}_{U})(2^{\gamma_{D}}-1)\sigma^2_{D}}{6(2^{\gamma_{A}}-1)\sigma^2_{A}N^2}
		-\frac{\kappa_{A}(|g|^2+|\mathbf{f}^H_{DU}\boldsymbol{\theta}^{**}_{U}|^2+|\mathbf{g}^H_{DU}\boldsymbol{\theta}^{**}_{D}|^2)}{N^2}
		\\& \geq
		\frac{\kappa_{U}(|h_{DA}|^2+|\mathbf{g}^H_{DA}\boldsymbol{\theta}^{**}_{D}|^2-2|h_{DA}||\mathbf{g}^H_{DA}\boldsymbol{\theta}^{**}_{D}|)(2^{\gamma_{A}}-1)}{6(2^{\gamma_{D}}-1)N^2} 
		\\&\quad+\frac{\kappa_{A}(|h_{AU}|^2+|\mathbf{f}^H_{AU}\boldsymbol{\theta}^{**}_{U}|^2-2|h_{AU}||\mathbf{f}^H_{AU}\boldsymbol{\theta}^{**}_{U}|)(2^{\gamma_{D}}-1)\sigma^2_{D}}{6(2^{\gamma_{A}}-1)\sigma^2_{A}N^2}
		\\&\quad-\frac{\kappa_{A}(|g|^2+|\mathbf{f}^H_{DU}\boldsymbol{\theta}^{**}_{U}|^2+|\mathbf{g}^H_{DU}\boldsymbol{\theta}^{**}_{D}|^2)}{N^2}=U(\boldsymbol{\theta}^{**}_{U}, \boldsymbol{\theta}^{**}_{D})-\frac{\kappa_{A}|g|^2}{N^2},
	\end{split}\label{eq:64}
\end{equation}
where the second inequality  holds   due to the fact that
\begin{equation}
	\lambda_{U}(\boldsymbol{\theta}^{**}_{U})\triangleq|h_{AU}+\mathbf{f}^H_{AU}\boldsymbol{\theta}^{**}_{U}|^2\geq(|h_{AU}|-|\mathbf{f}^H_{AU}\boldsymbol{\theta}^{**}_{U}|)^2=|h_{AU}|^2+|\mathbf{f}^H_{AU}\boldsymbol{\theta}^{**}_{U}|^2-2|h_{AU}||\mathbf{f}^H_{AU}\boldsymbol{\theta}^{**}_{U}|,
\end{equation}
\begin{equation}
	\lambda_{D}(\boldsymbol{\theta}^{**}_{D})\triangleq|h_{DA}+\mathbf{g}^H_{DA}\boldsymbol{\theta}^{**}_{D}|^2\geq(|h_{DA}|-|\mathbf{g}^H_{DA}\boldsymbol{\theta}^{**}_{D}|)^2=|h_{DA}|^2+|\mathbf{g}^H_{DA}\boldsymbol{\theta}^{**}_{D}|^2-2|h_{DA}||\mathbf{g}^H_{DA}\boldsymbol{\theta}^{**}_{D}|.
\end{equation}
Based on \eqref{eq:64},  we can guarantee $\tilde{L}_1(\boldsymbol{\theta}^{**}_{U}, \boldsymbol{\theta}^{**}_{D})\geq L_1(\boldsymbol{\theta}^{*}_{U}, \boldsymbol{\theta}^{*}_{D})$ if \eqref{conditionanalysis2} is satisfied.

Moreover,
assuming  Rayleigh fading channels for all IRS-related links,
due to \eqref{eq:suboptimal1111}, we have $|\mathbf{g}^H_{DA}\boldsymbol{\theta}^{**}_{D}|=\sum^{(1-\rho)N}_{n=1}|g_{IA,n}||g_{DI,n}|$ and $|\mathbf{f}^H_{AU}\boldsymbol{\theta}^{**}_{U}|=\sum^{\rho N}_{n=1}|f_{IU,n}||f_{AI,n}|$, where
$f_{IU,n}$ and $f_{AI,n}$ denote the $n$-th elements  of $\mathbf{f}_{IU}$ and $\mathbf{f}_{AI}$, respectively,
and $g_{IA,n}$ and $g_{DI,n}$ denote the $n$-th elements  of $\mathbf{g}_{IA}$ and $\mathbf{g}_{DI}$, respectively.
Since $|f_{IU,n}|$, $|f_{AI,n}|$, $|g_{IA,n}|$ and $|g_{DI,n}|$ are statistically independent and follow Rayleigh distribution with mean values
$\frac{\sqrt{\pi}\varrho_{f_{IU}}}{2}$, $\frac{\sqrt{\pi}\varrho_{f_{AI}}}{2}$, $\frac{\sqrt{\pi}\varrho_{g_{IA}}}{2}$, and $\frac{\sqrt{\pi}\varrho_{g_{DI}}}{2}$, respectively, we have $\mathbb{E}(|f_{IU,n}||f_{AI,n}|)=\frac{\pi\varrho_{f_{IU}}\varrho_{f_{AI}}}{4}$ and $\mathbb{E}(|g_{IA,n}||g_{DI,n}|)=\frac{\pi\varrho_{g_{IA}}\varrho_{g_{DI}}}{4}$.
When $N\rightarrow \infty$, we obtain
\begin{equation}
	\left|\frac{\mathbf{g}^H_{DA}\boldsymbol{\theta}^{**}_{D}}{(1-\rho)N}\right|=\frac{\sum^{(1-\rho)N}_{n=1}|f_{IU,n}||f_{AI,n}|}{(1-\rho)N}\rightarrow \frac{\pi\varrho_{f_{IU}}\varrho_{f_{AI}}}{4},
\end{equation}
\begin{equation}
	\left|\frac{\mathbf{f}^H_{AU}\boldsymbol{\theta}^{**}_{U}}{\rho N}\right|=\frac{\sum^{\rho N}_{n=1}|g_{IA,n}||g_{DI,n}|}{\rho N}\rightarrow \frac{\pi\varrho_{g_{IA}}\varrho_{g_{DI}}}{4}.
\end{equation}
Similarly, we  have
\begin{equation}
	|\frac{\mathbf{g}^H_{DA}\boldsymbol{\theta}^{**}_{D}}{N}|^2\rightarrow \frac{\pi^2\varrho^{2}_{f_{IU}}\varrho^{2}_{f_{AI}}(1-\rho)^2}{16}, \, \, 
	|\frac{\mathbf{f}^H_{AU}\boldsymbol{\theta}^{**}_{U}}{N}|^2\rightarrow \frac{\pi^2\varrho^{2}_{g_{IA}}\varrho^{2}_{g_{DI}}\rho^2}{16}.
\end{equation}
By leveraging the Lindeberg-L\'evy central limit theorem \cite{Probabilitydis} and the conclusion in \cite{QwuTWC2019},  we have $\mathbf{f}^H_{DU}\boldsymbol{\theta}^{**}_{U} \sim \mathcal{CN}(0, \rho N\varrho^2_{f_{IU}}\varrho^2_{f_{DI}})$ and $\mathbf{g}^H_{DU}\boldsymbol{\theta}^{**}_{D}\sim \mathcal{CN}(0, (1-\rho)N\varrho^2_{g_{IA}}\varrho^2_{g_{IU}})$ as $N\rightarrow \infty$,
and \\ $\frac{|\mathbf{f}^H_{DU}\boldsymbol{\theta}^{**}_{U}|^2+|\mathbf{g}^H_{DU}\boldsymbol{\theta}^{**}_{D}|^2}{N}\rightarrow \rho\varrho^2_{f_{IU}}\varrho^2_{f_{DI}}+(1-\rho)\varrho^2_{g_{IA}}\varrho^2_{g_{IU}}$. Thus, we obtain that when $N$ becomes asymptotically large, $U(\boldsymbol{\theta}^{**}_{U}, \boldsymbol{\theta}^{**}_{D})$ in this condition satisfies \eqref{condition22}. 
This thus  completes the proof.

\section{Proof for Theorem \ref{theorem_2}}
\label{appendixA}
By comparing $L_2(\boldsymbol{\theta}^*_{DU})$ and $\tilde{L}_2(\boldsymbol{\theta}^{**}_{DU})$, we obtain
\begin{equation}
\begin{split}
&\tilde{L}_2(\boldsymbol{\theta}^{**}_{DU})-L_2(\boldsymbol{\theta}^*_{DU})\\& =
\frac{\kappa_{U}(2^{\gamma_{A}}+1)(2^{\gamma_{A}}-1)\sigma^2_{A}}{2|h_{AU}|^2}+\frac{\kappa_{A}(2^{\gamma_{D}}+1)(2^{\gamma_{D}}-1)\sigma^2_{D}}{2\omega_{D}(\boldsymbol{\theta}^{**}_{DU})}\\&\quad-\left( \frac{\bar{\gamma}_3}{|h_{AU}|^2}+\frac{\bar{\gamma}_2}{\omega_{D}(\boldsymbol{\theta}^{*}_{DU})}+
\frac{\bar{\gamma}_1\omega_{DU}(\boldsymbol{\theta}^{*}_{DU})}{|h_{AU}|^2\omega_{D}(\boldsymbol{\theta}^{*}_{DU})}\right)
\\& \geq  \tilde{L}_2(\boldsymbol{\theta}^{**}_{DU})-L_2(\boldsymbol{\theta}^{**}_{DU})
\\&=\frac{\kappa_{U}(2^{\gamma_{A}}-1)^2\sigma^2_{A}}{2|h_{AU}|^2}+\frac{\kappa_{A}(2^{\gamma_{D}}-1)^2\sigma^2_{D}}{2\omega_{D}(\boldsymbol{\theta}^{**}_{DU})}-
\frac{\bar{\gamma}_1\omega_{DU}(\boldsymbol{\theta}^{**}_{DU})}{|h_{AU}|^2\omega_{D}(\boldsymbol{\theta}^{**}_{DU})}.
\end{split}
\end{equation}
Note that $\omega_{DU}(\boldsymbol{\theta}^{**}_{DU})\triangleq  |g+ \tilde{\mathbf{g}}^H_{DU}\boldsymbol{\theta}^{**}_{DU}|^2  \leq (|g|+|\tilde{\mathbf{g}}^H_{DU}\boldsymbol{\theta}^{**}_{DU}|)^2\leq 2(|g|^2+|\tilde{\mathbf{g}}^H_{DU}\boldsymbol{\theta}^{**}_{DU}|^2)$.
Thus, we have
\begin{equation}
\begin{split}
\tilde{L}_2(\boldsymbol{\theta}^{**}_{DU})-L_2(\boldsymbol{\theta}^*_{DU}) & \geq
\frac{\kappa_{U}(2^{\gamma_{A}}-1)^2\sigma^2_{A}}{2|h_{AU}|^2}+\frac{\kappa_{A}(2^{\gamma_{D}}-1)^2\sigma^2_{D}}{2\omega_{D}(\boldsymbol{\theta}^{**}_{DU})}-
\frac{2\bar{\gamma}_1(|g|^2+|\tilde{\mathbf{g}}^H_{DU}\boldsymbol{\theta}^{**}_{DU}|^2)}{|h_{AU}|^2\omega_{D}(\boldsymbol{\theta}^{**}_{DU})}.\label{eq:differentcase1}
\end{split}
\end{equation}
By multiplying $\frac{|h_{AU}|^2\omega_{D}(\boldsymbol{\theta}^{**}_{DU})}{2(2^{\gamma_{D}}-1)\sigma^2_{A}(2^{\gamma_{A}}-1)N^2}>0$ on both sides of \eqref{eq:differentcase1} and due to
the fact that $\omega_{D}(\boldsymbol{\theta}^{**}_{DU})\triangleq|h_{DA}+\tilde{\mathbf{g}}^H_{DA}\boldsymbol{\theta}^{**}_{DU}|^2\geq(|h_{DA}|-|\tilde{\mathbf{g}}^H_{DA}\boldsymbol{\theta}^{**}_{DU}|)^2=|h_{DA}|^2+|\tilde{\mathbf{g}}^H_{DA}\boldsymbol{\theta}^{**}_{DU}|^2-2|h_{DA}||\tilde{\mathbf{g}}^H_{DA}\boldsymbol{\theta}^{**}_{DU}|$, we obtain
\begin{equation}
\begin{split}
&\frac{|h_{AU}|^2\omega_{D}(\boldsymbol{\theta}^{**}_{DU})(\tilde{L}_2(\boldsymbol{\theta}^{**}_{DU})-L_2(\boldsymbol{\theta}^*_{DU}))}{2(2^{\gamma_{D}}-1)\sigma^2_{A}(2^{\gamma_{A}}-1)N^2} \\& \geq
\frac{\kappa_{U}\omega_{D}(\boldsymbol{\theta}^{**}_{DU})(2^{\gamma_{A}}-1)}{4(2^{\gamma_{D}}-1)N^2}+\frac{\kappa_{A}|h_{AU}|^2(2^{\gamma_{D}}-1)\sigma^2_{D}}{4(2^{\gamma_{A}}-1)\sigma^2_{A} N^2}-
\frac{\kappa_{A}(|g|^2+|\tilde{\mathbf{g}}^H_{DU}\boldsymbol{\theta}^{**}_{DU}|^2)}{N^2}
\\& \geq
\frac{\kappa_{U}(|h_{DA}|^2+|\tilde{\mathbf{g}}^H_{DA}\boldsymbol{\theta}^{**}_{DU}|^2-2|h_{DA}||\tilde{\mathbf{g}}^H_{DA}\boldsymbol{\theta}^{**}_{DU}|)(2^{\gamma_{A}}-1)}{4(2^{\gamma_{D}}-1)N^2}
+\frac{\kappa_{A}|h_{AU}|^2(2^{\gamma_{D}}-1)\sigma^2_{D}}{4(2^{\gamma_{A}}-1)\sigma^2_{A} N^2}
\\&\quad-\frac{\kappa_{A}(|g|^2+|\tilde{\mathbf{g}}^H_{DU}\boldsymbol{\theta}^{**}_{DU}|^2)}{N^2}=U(\boldsymbol{\theta}^{**}_{DU})-\frac{\kappa_{A}|g|^2}{N^2}.
\end{split}
\end{equation}
Therefore, we can guarantee that the power consumption of the FD system is lower
than that of the HD system in Case 2, i.e., $\tilde{L}_2(\boldsymbol{\theta}^{**}_{DU})\geq L_2(\boldsymbol{\theta}^*_{DU})$,  if the  condition \eqref{conditionanalysis2case2} is satisfied.

Assuming  Rayleigh fading channels for all IRS-related links,
when $N\rightarrow \infty$, similarly by leveraging the Lindeberg-L\'evy central limit theorem \cite{Probabilitydis},  we have  $\left|\frac{\tilde{\mathbf{g}}^H_{DA}\boldsymbol{\theta}^{**}_{DU}}{N}\right|^2\rightarrow \frac{\pi^2\varrho^2_{\tilde{g}_{IA}}\varrho^2_{\tilde{g}_{DI}}}{16}$,
$\frac{|\tilde{\mathbf{g}}^H_{DA}\boldsymbol{\theta}^{**}_{DU}|}{N^2}\rightarrow 0$,
and $\left|\frac{\tilde{\mathbf{g}}^H_{DU}\boldsymbol{\theta}^{**}_{DU}}{N}\right|^2\rightarrow 0$.
Thus, we obtain that when $N$ becomes asymptotically  large,
 $U(\boldsymbol{\theta}^{**}_{DU})$ in this condition satisfies \eqref{condition2}. The proof is thus completed.

\section{Proof for Theorem \ref{Theorem_3}}
\label{appendixBB}

Let us rewrite the expression of the minimum power consumption of the FD system in Case 1 as
\begin{equation}
L_1(\boldsymbol{\theta}^*_{U}, \boldsymbol{\theta}^*_{D})=\frac{\bar{\gamma}_1\lambda_{DU}(\boldsymbol{\theta}^*_{U}, \boldsymbol{\theta}^*_{D})}{\lambda_{D}(\boldsymbol{\theta}^*_{D})\lambda_{U}(\boldsymbol{\theta}^*_{U})}+\frac{\bar{\gamma}_2}{\lambda_{D}(\boldsymbol{\theta}^*_{D})}+\frac{\bar{\gamma}_3}{\lambda_{U}(\boldsymbol{\theta}^*_{U})}.\label{eq:418}
\end{equation}
 Based on the Cauchy-Schwarz inequality, we have
 \begin{equation}
 \lambda_{D}(\boldsymbol{\theta}^{**}_{D})\triangleq |h_{DA}+\mathbf{g}^H_{DA}\boldsymbol{\theta}^{**}_{D}|^2\geq (|h_{DA}|-|\mathbf{g}^H_{DA}\boldsymbol{\theta}^{**}_{D}|)^2=|h_{DA}|^2+|\mathbf{g}^H_{DA}\boldsymbol{\theta}^{**}_{D}|^2-2|h_{DA}||\mathbf{g}^H_{DA}\boldsymbol{\theta}^{**}_{D}|,
 \end{equation}
 \begin{equation}
 \lambda_{U}(\boldsymbol{\theta}^{**}_{U})\triangleq |h_{AU}+\mathbf{f}^H_{AU}\boldsymbol{\theta}^{**}_{U}|^2\geq (|h_{AU}|-|\mathbf{f}^H_{AU}\boldsymbol{\theta}^{**}_{U}|)^2=|h_{AU}|^2+|\mathbf{f}^H_{AU}\boldsymbol{\theta}^{**}_{U}|^2-2|h_{AU}||\mathbf{f}^H_{AU} \boldsymbol{\theta}^{**}_{U}|.
 \end{equation}
 Then, with the suboptimal solutions $\boldsymbol{\theta}^{**}_{U}$ and $\boldsymbol{\theta}^{**}_{D}$, we further obtain the following inequalities for the minimum power consumption:
\begin{equation}
\begin{split}
& L_1(\boldsymbol{\theta}^*_{U}, \boldsymbol{\theta}^*_{D}) \\& \leq \frac{\bar{\gamma}_1\lambda_{DU}(\boldsymbol{\theta}^{**}_{U}, \boldsymbol{\theta}^{**}_{D} )}{\lambda_{D}(\boldsymbol{\theta}^{**}_{D})\lambda_{U}(\boldsymbol{\theta}^{**}_{U})}+\frac{\bar{\gamma}_2}{\lambda_{D}(\boldsymbol{\theta}^{**}_{D})}+\frac{\bar{\gamma}_3}{\lambda_{U}(\boldsymbol{\theta}^{**}_{U})}
\\& \leq
\frac{\bar{\gamma}_1\lambda_{DU}(\boldsymbol{\theta}^{**}_{U}, \boldsymbol{\theta}^{**}_{D} )}{(|h_{AU}|-|\mathbf{f}^H_{AU}\boldsymbol{\theta}^{**}_{U}|)^2(|h_{DA}|-|\mathbf{g}^H_{DA}\boldsymbol{\theta}^{**}_{D}|)^2}+\frac{\bar{\gamma}_2}{(|h_{DA}|-|\mathbf{g}^H_{DA}\boldsymbol{\theta}^{**}_{D}|)^2}+\frac{\bar{\gamma}_3}{(|h_{AU}|-|\mathbf{f}^H_{AU}\boldsymbol{\theta}^{**}_{U}|)^2}
\\& \leq \tilde{U}(\boldsymbol{\theta}^{**}_{U}, \boldsymbol{\theta}^{**}_{D})
\\& \triangleq
\frac{3\bar{\gamma}_1(|g|^2+|\mathbf{f}^H_{DU}\boldsymbol{\theta}^{**}_{U}|^2+|\mathbf{g}^H_{DU}\boldsymbol{\theta}^{**}_{D}|^2)}{(|h_{AU}|^2+|\mathbf{f}^H_{AU}\boldsymbol{\theta}^{**}_{U}|^2-2|h_{AU}||\mathbf{f}^H_{AU} \boldsymbol{\theta}^{**}_{U}|)(|h_{DA}|^2+|\mathbf{g}^H_{DA}\boldsymbol{\theta}^{**}_{D}|^2-2|h_{DA}||\mathbf{g}^H_{DA}\boldsymbol{\theta}^{**}_{D}|)}\\&\quad+\frac{\bar{\gamma}_2}{|h_{DA}|^2+|\mathbf{g}^H_{DA}\boldsymbol{\theta}^{**}_{D}|^2-2|h_{DA}||\mathbf{g}^H_{DA}\boldsymbol{\theta}^{**}_{D}|}+\frac{\bar{\gamma}_3}{|h_{AU}|^2+|\mathbf{f}^H_{AU}\boldsymbol{\theta}^{**}_{U}|^2-2|h_{AU}||\mathbf{f}^H_{AU}\boldsymbol{\theta}^{**}_{U}|},
\end{split}
\label{powerconsumption3}
\end{equation}
where the last inequality holds due to the fact that
\begin{equation}
\begin{split}
\lambda_{DU}(\boldsymbol{\theta}^{**}_{U}, \boldsymbol{\theta}^{**}_{D} )&\triangleq |g+\mathbf{f}^H_{DU}\boldsymbol{\theta}^{**}_{U}+\mathbf{g}^H_{DU}\boldsymbol{\theta}^{**}_{D}|^2\leq (|g|+|\mathbf{f}^H_{DU}\boldsymbol{\theta}^{**}_{U}|+|\mathbf{g}^H_{DU}\boldsymbol{\theta}^{**}_{D}|)^2\\&\leq 3(|g|^2+|\mathbf{f}^H_{DU}\boldsymbol{\theta}^{**}_{U}|^2+|\mathbf{g}^H_{DU}\boldsymbol{\theta}^{**}_{D}|^2).
\end{split}
\end{equation}
Furthermore, let us rewrite  $\tilde{U}(\boldsymbol{\theta}^{**}_{U}, \boldsymbol{\theta}^{**}_{D})$ as
\begin{equation}
\begin{split}
&\tilde{U}(\boldsymbol{\theta}^{**}_{U}, \boldsymbol{\theta}^{**}_{D})\\&\triangleq
\frac{3\bar{\gamma}_1(|g|^2+N\frac{(|\mathbf{f}^H_{DU}\boldsymbol{\theta}^{**}_{U}|^2+|\mathbf{g}^H_{DU}\boldsymbol{\theta}^{**}_{D}|^2)}{N})}{(|h_{DA}|^2+N^2\frac{|\mathbf{g}^H_{DA}\boldsymbol{\theta}^{**}_{D}|^2}{N^2}-2N|h_{DA}|\frac{|\mathbf{g}^H_{DA}\boldsymbol{\theta}^{**}_{D}|}{N})(|h_{AU}|^2+N^2\frac{|\mathbf{f}^H_{AU}\boldsymbol{\theta}^{**}_{U}|^2}{N^2}-2N|h_{AU}|\frac{|\mathbf{f}^H_{AU} \boldsymbol{\theta}^{**}_{U}|}{N})}\\&\quad+\frac{\bar{\gamma}_2}{|h_{DA}|^2+N^2 \frac{|\mathbf{g}^H_{DA}\boldsymbol{\theta}^{**}_{D}|^2}{N^2}-2N|h_{DA}|\frac{|\mathbf{g}^H_{DA}\boldsymbol{\theta}^{**}_{D}|}{N}}+\frac{\bar{\gamma}_3}{|h_{AU}|^2+N^2 \frac{|\mathbf{f}^H_{AU}\boldsymbol{\theta}^{**}_{U}|^2}{N^2}-2N|h_{AU}|\frac{|\mathbf{f}^H_{AU}\boldsymbol{\theta}^{**}_{U}|}{N}}.
\end{split}
\end{equation}
Assuming Rayleigh fading channels for all IRS-related links, similarly when $N\rightarrow \infty$ we obtain
\begin{equation}
\begin{split}
&|\frac{\mathbf{g}^H_{DA}\boldsymbol{\theta}^{**}_{D}}{N}|\rightarrow \frac{\pi\varrho_{g_{IA}}\varrho_{g_{DI}}(1-\rho)}{4}, \, |\frac{\mathbf{g}^H_{DA}\boldsymbol{\theta}^{**}_{D}}{N}|^2\rightarrow  \frac{\pi^2\varrho^2_{g_{IA}}\varrho^2_{g_{DI}}(1-\rho)^2}{16}, \,
|\frac{\mathbf{f}^H_{AU}\boldsymbol{\theta}^{**}_{U}}{N}|\rightarrow \frac{\pi\varrho_{f_{IU}}\varrho_{f_{AI}}\rho}{4},  \\&
|\frac{\mathbf{f}^H_{AU}\boldsymbol{\theta}^{**}_{U}}{N}|^2\rightarrow  \frac{\pi^2\varrho^2_{f_{IU}}\varrho^2_{f_{AI}}\rho^2}{16}, \,
 \frac{|\mathbf{f}^H_{DU}\boldsymbol{\theta}^{**}_{U}|^2+|\mathbf{g}^H_{DU}\boldsymbol{\theta}^{**}_{D}|^2}{N}\rightarrow \rho\varrho^2_{f_{IU}}\varrho^2_{f_{DI}}+(1-\rho)\varrho^2_{g_{DI}}\varrho^2_{g_{IU}}.
\end{split}
\end{equation}
Therefore, when $N$ becomes asymptotically  large, the upper bound of $L_3(\boldsymbol{\theta}^*_{U}, \boldsymbol{\theta}^*_{D})$ satisfies
\begin{equation}
\begin{split}
\tilde{U}(\boldsymbol{\theta}^{**}_{U}, \boldsymbol{\theta}^{**}_{D})&\rightarrow
\frac{3\bar{\gamma}_1(|g|^2+N(\rho\varrho^2_{f_{IU}}\varrho^2_{f_{DI}}+(1-\rho)\varrho^2_{g_{DI}}\varrho^2_{g_{IU}}))}{(|h_{DA}|^2+\frac{\varpi^2_{g}}{16}N^2-\frac{|h_{DA}|\varpi_{g}}{2}N)(|h_{AU}|^2+\frac{\varpi^2_{f}}{16}N^2-\frac{|h_{AU}|\varpi_{f}}{2}N)}
\\&\quad+\frac{\bar{\gamma}_2}{|h_{DA}|^2+\frac{\varpi^2_{g}}{16}N^2-\frac{|h_{DA}|\varpi_{g}}{2}N} 
+\frac{\bar{\gamma}_3}{|h_{AU}|^2+\frac{\varpi^2_{f}}{16}N^2-\frac{|h_{AU}|\varpi_{f}}{2}N}\rightarrow 0,\label{powerconsumption4}
\end{split}
\end{equation}
where $\varpi_{g}\triangleq \pi\varrho_{g_{IA}}\varrho_{g_{DI}}(1-\rho)$ and $\varpi_{f}\triangleq \pi\varrho_{f_{IU}}\varrho_{f_{AI}}\rho$. Since the last two terms on the right hand side (RHS)  of \eqref{powerconsumption4} dominate in the magnitude of $\tilde{U}(\boldsymbol{\theta}^{**}_{U}, \boldsymbol{\theta}^{**}_{D})$,
we can see that $\tilde{U}(\boldsymbol{\theta}^{**}_{U}, \boldsymbol{\theta}^{**}_{D})$ decreases quadratically with the increasing of $N$. The proof is thus completed.

\section{ Proof of equivalence between  \eqref{eq_problemtheta2222} and \eqref{trans_prob}}
\label{appendixB}
By fixing the other variables, the optimum $v$ for minimizing \eqref{trans_prob} is given by
  \begin{equation}
    v=\frac{|h_{DA}+\mathbf{g}^H_{DA}\boldsymbol{\theta}_{D}||h_{AU}+\mathbf{f}^H_{AU}\boldsymbol{\theta}_{U}|}{\bar{\gamma}_1|g+\mathbf{f}^H_{DU}\boldsymbol{\theta}_{U}+\mathbf{g}^H_{DU}\boldsymbol{\theta}_{D}|^2 + \bar{\gamma}_2|h_{AU}+\mathbf{f}^H_{AU}\boldsymbol{\theta}_{U}|^2 + \bar{\gamma}_3|h_{DA}+\mathbf{g}^H_{DA}\boldsymbol{\theta}_{D}|^2}. \label{appendixCCCCC1}
    \end{equation} 
By substituting \eqref{appendixCCCCC1}
 into \eqref{trans_prob}, we have the following equivalent optimization problem:
 \begin{equation} 
\max_{\boldsymbol{\phi}_{U}, \boldsymbol{\phi}_{D}}\quad
\frac{|h_{DA}+\mathbf{g}^H_{DA}\boldsymbol{\psi}_{D}|^2|h_{AU}+\mathbf{f}^H_{AU}\boldsymbol{\psi}_{U}|^2}{\bar{\gamma}_1|g+\mathbf{f}^H_{DU}\boldsymbol{\psi}_{U}+\mathbf{g}^H_{DU}\boldsymbol{\psi}_{D}|^2 + \bar{\gamma}_2|h_{AU}+\mathbf{f}^H_{AU}\boldsymbol{\psi}_{U}|^2 + \bar{\gamma}_3|h_{DA}+\mathbf{g}^H_{DA}\boldsymbol{\psi}_{D}|^2}. 
\label{eq_problemtheta2222appendix}
\end{equation}
It is readily seen that \eqref{eq_problemtheta2222}  is equivalent to \eqref{trans_prob}. 
This thus completes the proof.
\end{appendices}

\bibliographystyle{IEEETran}
\bibliography{references}

\end{document}